\documentclass[DIV=12,11pt]{scrartcl}

\usepackage{amsthm,amssymb,amsmath}
\usepackage[noblocks]{authblk}

\setkomafont{title}{\normalfont \LARGE \bfseries}
\setkomafont{section}{\normalfont \Large \bfseries \boldmath}
\setkomafont{subsection}{\normalfont \large \bfseries}
\setkomafont{subsubsection}{\normalfont \normalsize \bfseries}
\setkomafont{descriptionlabel}{\itshape}
\setkomafont{paragraph}{\normalfont \bfseries \boldmath}

\newtheorem{theorem}{Theorem}
\newtheorem{lemma}[theorem]{Lemma}
\newtheorem{corollary}[theorem]{Corollary}

\usepackage{graphicx}
\usepackage{hyperref}
\usepackage[numbers]{natbib}
\usepackage{url}
\usepackage{color}

\newcommand{\Cov}{\mathop{\mathgroup\symoperators Cov}\nolimits}
\newcommand{\Exp}{\mathop{\mathgroup\symoperators Exp}\nolimits}
\newcommand{\Var}{\mathop{\mathgroup\symoperators Var}\nolimits}

\newcommand{\blokje}{}

\begin{document}

\title{Probabilistic Analysis of Facility Location on Random Shortest Path Metrics\footnote{An extended abstract of this work will appear in the \textit{Proceedings of the 15th Conference on Computability in Europe (CiE)}.}}

\author{Stefan Klootwijk}
\author{Bodo Manthey}

\affil{University of Twente, Enschede, The Netherlands, \texttt{\{s.klootwijk,b.manthey\}@utwente.nl}}

\maketitle

\begin{abstract}
The facility location problem is an $\mathcal{NP}$-hard optimization problem. Therefore, approximation algorithms are often used to solve large instances. Such algorithms often perform much better than worst-case analysis suggests. Therefore, probabilistic analysis is a widely used tool to analyze such algorithms. Most research on probabilistic analysis of $\mathcal{NP}$-hard optimization problems involving metric spaces, such as the facility location problem, has been focused on Euclidean instances, and also instances with independent (random) edge lengths, which are non-metric, have been researched. We would like to extend this knowledge to other, more general, metrics.

We investigate the facility location problem using random shortest path metrics. We analyze some probabilistic properties for a simple greedy heuristic which gives a solution to the facility location problem: opening the $\kappa$ cheapest facilities (with $\kappa$ only depending on the facility opening costs). If the facility opening costs are such that $\kappa$ is not too large, then we show that this heuristic is asymptotically optimal. On the other hand, for large values of $\kappa$, the analysis becomes more difficult, and we provide a closed-form expression as upper bound for the expected approximation ratio. In the special case where all facility opening costs are equal this closed-form expression reduces to $O(\sqrt[4]{\ln(n)})$ or $O(1)$ or even $1+o(1)$ if the opening costs are sufficiently small.

\end{abstract}

\nocite{Ahn1988}
\nocite{Aven1985}
\nocite{Bringmann2015}
\nocite{Cornuejols1990}
\nocite{Davis1993}
\nocite{Flaxman2007}
\nocite{Frieze2007}
\nocite{Hammersley1965}
\nocite{Hassin1985}
\nocite{Howard2004}
\nocite{Janson1999}
\nocite{Karp1985}
\nocite{Li2011}
\nocite{Luke1969}
\nocite{Nagaraja2006}
\nocite{Renyi1953}
\nocite{Ross2010}
\nocite{Shaked2007}

\section{Introduction}

Large-scale combinatorial optimization problems, such as the facility location problem, show up in many applications. These problems become computationally intractable as the instances grow. This issue is often tackled by (successfully) using approximation algorithms or ad-hoc heuristics to solve these optimization problems. In practical situations these, often simple, heuristics have a remarkable performance, even though theoretical results about them are way more pessimistic.

Over the last decades, probabilistic analysis has become an important tool to explain this difference. One of the main challenges here is to come up with a good probabilistic model for generating instances: this model should reflect realistic instances, but it should also be sufficiently simple in order to make the probabilistic analysis possible.

Until recently, in almost all cases either instances with independent edge lengths, or instances with Euclidean distances have been used for this purpose~\cite{Ahn1988,Frieze2007}. These models are indeed sufficiently simple, but they have shortcomings with respect to reflecting realistic instances: realistic instances are often metric, although not Euclidean, and the independent edge lengths do not even yield a metric space.

In order to overcome this, Bringmann et al.~\cite{Bringmann2015} used the following model for generating random metric spaces, which had been proposed by Karp and Steele~\cite{Karp1985}. Given an undirected complete graph, start by drawing random edge weights for each edge independently and then define the distance between any two vertices as the total weight of the shortest path between them, measured with respect to the random weights. Bringmann et al. called this model \emph{random shortest path metrics}. This model is also known as \emph{first-passage percolation}, introduced by Hammersley and Welsh as a model for fluid flow through a (random) porous medium~\cite{Hammersley1965,Howard2004}.

\subsection{Related Work}

Although a lot of studies have been conducted on random shortest path metrics, or first-passage percolation (e.g. \cite{Davis1993,Hassin1985,Janson1999}), systematic research of the behavior of (simple) heuristics and approximation algorithms for optimization problems on random shortest path metrics was initiated only
recently~\cite{Bringmann2015}. They provide some structural properties of random shortest path metrics, including the existence of a good clustering. These properties are then used for a probabilistic analysis of simple algorithms for several optimization problems, including the minimum-weight perfect matching problem and the $k$-median problem.

For the facility location problem, several sophisticated polynomial-time approximation algorithms exist, the best one currently having a worst-case approximation ratio of $1.488$ \cite{Li2011}. Flaxman et al. conducted a probabilistic analysis for the facility location problem using Euclidean distances \cite{Flaxman2007}. They expected to show that some polynomial-time approximation algorithms would be asymptotically optimal under these circumstances, but found out that this is not the case. On the other hand, they described a trivial heuristic which is asymptotically optimal in the Euclidean model.

\subsection{Our Results}

This paper aims at extending our knowledge about the probabilistic behavior of (simple) heuristics and approximation algorithms for optimization problems using random shortest path metrics. We will do so by investigating the probabilistic properties of a rather simple heuristic for the facility location problem, which opens the $\kappa$ cheapest facilities (breaking ties arbitrarily) where $\kappa$ only depends on the facility opening costs. Due to the simple structure of this heuristic, our results are more structural than algorithmic in nature.

We show that this heuristic yields a $1+o(1)$ approximation ratio in expectation if the facility opening costs are such that $\kappa\in o(n)$. For $\kappa\in\Theta(n)$ the analysis becomes more difficult, and we provide a closed-form expression as upper bound for the expected approximation ratio. We will also show that this closed-form expression is $O(\sqrt[4]{\ln(n)})$ if all facility opening costs are equal. This can be improved to $O(1)$ or even $1+o(1)$ when the facility opening costs are sufficiently small. Note that we will focus on the expected approximation ratio and not on the ratio of expectations, since a disadvantage of the latter is that it does not directly compare the performance of the heuristic on specific instances.

We start by giving a mathematical description of random shortest path metrics and the facility location problem (Section~\ref{sect:modelnotation}). After that, we introduce our simple heuristic properly and have a brief look at its behavior (Section~\ref{sect:heur}). Then we present some general technical results (Section~\ref{sect:prelim}) and two different bounds for the optimal solution (Section~\ref{sect:OPTbounds}) that we will use to prove our main results in Section~\ref{sect:probanalysis}. We conclude with some final remarks (Section~\ref{sect:final}).


\section{Notation and Model}\label{sect:modelnotation}

In this paper, we use $X\sim P$ to denote that a random variable $X$ is distributed using a probability distribution $P$. $\Exp(\lambda)$ is being used to denote the exponential distribution with parameter $\lambda$. In particular, we use $X\sim\sum_{i=1}^n\Exp(\lambda_i)$ to denote that $X$ is the sum of $n$ independent exponentially distributed random variables with parameters $\lambda_1,\ldots,\lambda_n$. If $\lambda_1=\ldots=\lambda_n=\lambda$, then $X$ is a Gamma distributed random variable with parameters $n$ and $\lambda$, denoted by $X\sim\Gamma(n,\lambda)$.

For $n\in\mathbb{N}$, we use $[n]$ as shorthand notation for $\{1,\ldots,n\}$. If $X_1,\ldots,X_m$ are $m$ random variables, then $X_{(1)},\ldots,X_{(m)}$ are the order statistics corresponding to $X_1,\ldots,X_m$ if $X_{(i)}$ is the $i$th smallest value among $X_1,\ldots,X_m$ for all $i\in[m]$. Furthermore we use $H_n$ as shorthand notation for the $n$th harmonic number, i.e., $H_n=\sum_{i=1}^n1/i$. Finally, if a random variable $X$ is stochastically dominated by a random variable $Y$, i.e., we have $F_X(x)\geq F_Y(x)$ for all $x$ (where $X\sim F_X$ and $Y\sim F_Y$), we denote this by $X\precsim Y$.

\paragraph{Random Shortest Path Metrics.}

Given an undirected complete graph $G=(V,E)$ on $n$ vertices, we construct the corresponding random shortest path metric as follows. First, for each edge $e\in E$, we draw a random edge weight $w(e)$ independently from an exponential distribution\footnote{Exponential distributions are technically easiest to handle due to their memorylessness property. A (continuous, non-negative) probability distribution of a random variable $X$ is said to be memoryless if and only if $\mathbb{P}(X>s+t\mid X>t)=\mathbb{P}(X>s)$ for all $s,t\geq0$.~\cite[p.~294]{Ross2010}} with parameter 1. Given these random edge weights $w(e)$, the distance $d(u,v)$ between each pair of vertices $u,v\in V$ is defined as the minimum total weight of a $u,v$-path in $G$. 
Note that this definition yields the following properties: $d(v,v)=0$ for all $v\in V$, $d(u,v)=d(v,u)$ for all $u,v\in V$, and $d(u,v)\leq d(u,s)+d(s,v)$ for all $u,s,v\in V$. We call the complete graph with distances $d$ obtained from this process a random shortest path metric.

\paragraph{Facility Location Problem.}

We consider the (uncapacitated) facility location problem, in which we are given a complete undirected graph $G=(V,E)$ on $n$ vertices, distances $d:V\times V\to\mathbb{R}_{\geq0}$ between each pair of vertices, and opening costs $f:V\to\mathbb{R}_{>0}$. In this paper, the distances are randomly generated, according to the random shortest path metric described above. Moreover, w.l.o.g. we assume that the vertices are numbered in such a way that the opening costs satisfy $f_1\leq f_2\leq\ldots\leq f_n$ and we assume that these costs are predetermined, independent of the random edge weights. We will use $F_k$ as a shorthand notation for $\sum_{i=1}^kf_i$. Additionally, we assume that the ratios between the opening costs are polynomially bounded, i.e., we assume $f_n/f_1\leq n^q$ for some constant $q$ as $n\to\infty$.

The goal of the facility location problem is to find a nonempty subset $U\subseteq V$ such that the total cost $c(U):=f(U)+\sum_{v\in V}\min_{u\in U}d(u,v)$ is minimal, where $f(U)$ denotes the total opening cost of all facilities in $U$. This problem is $\mathcal{NP}$-hard~\cite{Cornuejols1990}. We use $\mathsf{OPT}$ to denote the total cost of an optimal solution, i.e.,
\begin{equation*}
\mathsf{OPT}=\min_{\varnothing\neq U\subseteq V}c(U).
\end{equation*}
One of the tools we use in our proofs in Section~\ref{sect:probanalysis} involves fixing the number of facilities that has to be opened. We use $\mathsf{OPT}_k$ to denote the total cost of the best solution to the facility location problem with the additional constraint that exactly $k$ facilities need to be opened, i.e.,
\begin{equation*}
\mathsf{OPT}_k=\underset{|U|=k}{\min_{\varnothing\neq U\subseteq V}}c(U).
\end{equation*}
Note that $\mathsf{OPT}=\min_{k\in[n]}\mathsf{OPT}_k$ by these definitions.

\section{A simple heuristic and some of its properties}\label{sect:heur}

In this paper we are interested in a rather simple heuristic that only takes the facility opening costs $f_i$ into account while determining which facilities to open and which not, independently of the metric space. Define $\kappa:=\kappa(n;f_1,\ldots,f_n)=\max\{i\in[n]:f_i<1/(i-1)\}$. Then our heuristic opens the $\kappa$ cheapest facilities (breaking ties arbitrarily). Note that in the special case where all opening costs are the same, i.e. $f_1=\ldots=f_n=f$, this corresponds to $\kappa=\min\{\lceil1/f\rceil,n\}$.

This rather particular value of $\kappa$ is originates from the following intuitive argument. Based on the results of Bringmann et al.~\cite[Lemma~5.1]{Bringmann2015} (see below) we know that the expected cost of the solution that opens the $k$ cheapest facilities is given by $g(k):=F_k+H_{n-1}-H_{k-1}$. This convex function decreases as long as $k$ satisfies $f_k<1/(k-1)$. Therefore, at least intuitively, the value of $\kappa$ that we use is likely to provide a relatively `good' solution.

We will show that this is indeed the case. Our main result will be split into two parts, based on the actual value of $\kappa$. If $\kappa\in o(n)$ (i.e. if there are `many' relatively expensive facilities), then we will show that our simple heuristic is asymptotically optimal for any polynomially bounded opening costs (that satisfy $\kappa\in o(n)$). On the other hand, if $\kappa\in\Theta(n)$, then the analysis becomes more difficult, and we will only provide a closed-form expression that can be used to determine an upper bound for the expected approximation ratio. We will show that this expression yields an $O(\sqrt[4]{\ln(n)})$ approximation ratio in the special case with $f_1=\ldots=f_n=f$, and $O(1)$ or even $1+o(1)$ if $f$ is sufficiently small.

Throughout the remainder of this paper we will use $\mathsf{ALG}$ to denote the value of the solution provided by this heuristic.

\paragraph{Probability distribution of $\mathsf{ALG}$.}\label{par:ALGdist}

In this section we derive the probability distribution of the value of the solution provided by our simple greedy heuristic, $\mathsf{ALG}$, and derive its expectation.

If $\kappa=n$, then $\mathsf{ALG}$ denotes the cost of the solution which opens a facility at every vertex $v\in V$. So, we have $\mathsf{ALG}=F_n$, and, in particular, $\mathbb{P}(\mathsf{ALG}=F_n)=1$.

If $1\leq\kappa<n$, then the distribution of $\mathsf{ALG}$ is less trivial. In this case, the total opening costs are given by $F_\kappa$, whereas, the distribution of the connection costs is known and given by $\sum_{i=\kappa}^{n-1}\Exp(i)$ \cite[Sect.~5]{Bringmann2015}. This results in $\mathsf{ALG}-F_\kappa\sim\sum_{i=\kappa}^{n-1}\Exp(i)$.

Using this probability distribution, we can derive the expected value of $\mathsf{ALG}$. If $\kappa=n$, then it follows trivially that $\mathbb{E}[\mathsf{ALG}]=F_n$. If $1\leq\kappa<n$, then we have
\begin{equation*}
\mathbb{E}[\mathsf{ALG}]=F_\kappa+\sum_{i=\kappa}^{n-1}\frac1i=F_\kappa+H_{n-1}-H_{\kappa-1}=F_\kappa+\ln(n/\kappa)+\Theta(1).
\end{equation*}

\section{Technical observations}\label{sect:prelim}

In this section we present some technical lemmas that are being used for the proofs of our theorems in Section~\ref{sect:probanalysis}. These lemmas do not provide new structural insights, but are nonetheless very helpful for our proofs.

First of all, we will use the Cauchy-Schwarz inequality to bound the expected approximation ratio of our simple greedy heuristic. For general random variables $X$, $Y$, this inequality states that $|\mathbb{E}[XY]|\leq\sqrt{\mathbb{E}[X^2]\mathbb{E}[Y^2]}$.

Secondly, we will bound a sum of exponential distributions by a Gamma distribution. The following Lemma enables us to do so.
\begin{lemma}[{\cite[Ex.~1.A.24]{Shaked2007}}]
	\label{lemma:expGam}
	Let $X_i\sim\Exp(\lambda_i)$ independently, $i=1,\ldots,m$. Moreover, let $Y_i\sim\Exp(\eta)$ independently, $i=1,\ldots,m$. Then we have
	\begin{equation*}
	\sum_{i=1}^mX_i\succsim\sum_{i=1}^mY_i\qquad\text{if and only if}\qquad\prod_{i=1}^m\lambda_i\leq\eta^m.
	\end{equation*}
\end{lemma}
We will use the following upper bound for the expectation of the maximum of a number of (dependent) random variables.
\begin{lemma}[{\cite[Thm.~2.1]{Aven1985}}]
	\label{lemma:expecmax}
	Let $X_1,\ldots,X_n$ be a sequence of random variables, each with finite mean and variance. Then it follows that
	\begin{equation*}
	\mathbb{E}\left[\max_iX_i\right]\leq\max_i\mathbb{E}\left[X_i\right]+\sqrt{\frac{n-1}{n}\cdot\sum_{i=1}^n\Var(X_i)}.
	\end{equation*}
\end{lemma}
We will also make use of R\'enyi's representation \cite{Nagaraja2006,Renyi1953} in order to be able to link sums and order statistics of exponentially distributed random variables. It states the following.
\begin{lemma}
	\label{lemma:renyi}
	Let $X_i\sim\Exp(\lambda)$ independently, $i=1,\ldots,m$, and let $X_{(1)},\ldots,$ $X_{(m)}$ be the order statistics corresponding to $X_1,\ldots,X_m$. Then, for any $i\in[m]$,
	\begin{equation*}
	X_{(i)}=\frac1\lambda\sum_{j=1}^i\frac{Z_j}{m-j+1},
	\end{equation*}
	where $Z_j\sim\Exp(1)$ independently, and where ``$=$'' means equal distribution.
\end{lemma}
A special case of R\'enyi's representation is given by the following corollary.
\begin{corollary}
	\label{cor:renyi}
	Let $Y_i\sim\Exp(1)$ independently, $i=1,\ldots,n-1$, and let $Y_{(1)},\ldots,Y_{(n-1)}$ be the order statistics corresponding to $Y_1,\ldots,Y_{n-1}$. Then, for any $i\in[n-1]$,
	\begin{equation*}
	Y_{(n-i)}\sim\sum_{k=i}^{n-1}\Exp(k).
	\end{equation*}
\end{corollary}
\begin{proof}
	Let $Z_j\sim\Exp(1)$ independently. Using Lemma~\ref{lemma:renyi} it follows immediately that
	\begin{equation*}
	Y_{(n-i)}=\sum_{j=1}^{n-i}\frac{Z_j}{n-j}=\sum_{k=i}^{n-1}\frac{Z_{n-k}}{k}\sim\sum_{k=1}^{n-1}\Exp(k),
	\end{equation*}
	since $\Exp(1)/k\sim\Exp(k)$.\blokje
\end{proof}
Moreover, we use the following bound for the expected value of the ratio $X/Y$ for two dependent nonnegative variables $X$ and $Y$, conditioned on the event that $Y$ is relatively small.
\begin{lemma}
	\label{lemma:E[X/Y]}
	Let $X$ and $Y$ be two arbitrary nonnegative random variables and assume that $\mathbb{P}(Y\leq\delta)=0$ for some $\delta>0$. Then, for any $y$ that satisfies $\mathbb{P}(Y<y)>0$, we have
	\begin{equation*}
	\mathbb{P}(Y<y)\cdot\mathbb{E}\left[\frac{X}{Y}\;\middle|\;Y<y\right]\leq\frac1{\delta^2}\cdot\mathbb{P}(Y<y)+\int_{1/\delta^2}^{\infty}\mathbb{P}(X\geq\sqrt{x})\,\mathrm{d}x.
	\end{equation*}
\end{lemma}
\begin{proof}
	The expected value on the left-hand side can be computed and bounded as follows:
	\begin{align*}
	\mathbb{P}(Y<y)\cdot\mathbb{E}\left[\frac{X}{Y}\;\middle|\;Y<y\right]&=\mathbb{P}(Y<y)\cdot\int_0^\infty\mathbb{P}\left(\frac{X}{Y}\geq x\;\middle|\;Y<y\right)\,\mathrm{d}x\\
	&\leq\mathbb{P}(Y<y)\cdot\left(\frac1{\delta^2}+\int_{1/\delta^2}^\infty\mathbb{P}\left(\frac{X}{Y}\geq x\;\middle|\;Y<y\right)\,\mathrm{d}x\right)\\
	&=\frac1{\delta^2}\cdot\mathbb{P}(Y<y)+\int_{1/\delta^2}^\infty\mathbb{P}\left(\frac{X}{Y}\geq x\text{ and }Y<y\right)\,\mathrm{d}x\\
	&\leq\frac1{\delta^2}\cdot\mathbb{P}(Y<y)+\int_{1/\delta^2}^\infty\mathbb{P}\left(\frac{X}{Y}\geq x\right)\,\mathrm{d}x.
	\end{align*}
	Observe that $X/Y\geq x$ implies $X\geq\sqrt{x}$ or $Y\leq1/\sqrt{x}$. This observation yields
	\begin{align*}
	\mathbb{P}(Y<y)\cdot\mathbb{E}\left[\frac{X}{Y}\;\middle|\;Y<y\right]&\leq\frac1{\delta^2}\cdot\mathbb{P}(Y<y)+\int_{1/\delta^2}^\infty\mathbb{P}\left(X\geq\sqrt{x}\text{ or }Y\leq\frac1{\sqrt{x}}\right)\,\mathrm{d}x\\
	&\leq\frac1{\delta^2}\cdot\mathbb{P}(Y<y)+\int_{1/\delta^2}^\infty\mathbb{P}\left(X\geq\sqrt{x}\right)\,\mathrm{d}x+\int_{1/\delta^2}^\infty\mathbb{P}\left(Y\leq\frac1{\sqrt{x}}\right)\,\mathrm{d}x.
	\end{align*}
	Since $\mathbb{P}(Y\leq\delta)=0$, the second integral vanishes, which leaves us with the desired result.\blokje
\end{proof}

\section{Bounds for the optimal solution}\label{sect:OPTbounds}

Not much is known about the distribution of the value of the optimal solution, $\mathsf{OPT}$, and about the distributions of $\mathsf{OPT}_k$. Therefore, in this section we derive two bounds for these optimal solutions which we can use in Section \ref{sect:probanalysis}.

We start with an upper bound for the cumulative distribution function of $\mathsf{OPT}$ that works good for relative small values of $\mathsf{OPT}$ (i.e. values close to $F_1$).
\begin{lemma}
	\label{lemma:OPTsmall}
	Let $z\in[F_1,F_n]$ and define $\zeta:=\max\{k:z\geq F_k\}$. Then, for any given opening costs $f_i$, we have
	\begin{equation*}
	\mathbb{P}(\mathsf{OPT}<z)\leq\sum_{i=1}^{\zeta}\binom{n}{i}\binom{n-1}{i-1}\left(1-e^{-(z-F_i)}\right)^{n-i}.
	\end{equation*}
\end{lemma}
\begin{proof}
	Let $L$ denote the number of open facilities in the optimal solution (if there are multiple optimal solutions, pick one arbitrarily). If $\mathsf{OPT}<z$, then we know that $L=i$ for some $i\in[\zeta]$. Since these cases are disjoint, we can condition as follows:
	\begin{equation*}
	\mathbb{P}(\mathsf{OPT}<z)=\sum_{i=1}^{\zeta}\mathbb{P}(\mathsf{OPT}<z\mid L=i)\cdot\mathbb{P}(L=i)\leq\sum_{i=1}^{\zeta}\mathbb{P}(\mathsf{OPT}<z\mid L=i).
	\end{equation*}
	Recall that $f(U)=\sum_{j\in U}f_j$ is the total opening cost of all facilities in $U$. Using the union bound, we can derive that
	\begin{align*}
	\mathbb{P}(\mathsf{OPT}<z\mid L=i)&=\mathbb{P}\left(\exists U\subseteq V,\,|U|=i\,:\,c(U)<z\right)\\
	&\leq\mathbb{P}\left(\exists U\subseteq V,\,|U|=i\,:\,c(U)-f(U)<z-F_i\right)\\
	&\leq\binom{n}i\cdot\mathbb{P}\left(\sum_{k=i}^{n-1}\Exp(k)<z-F_i\right),
	\end{align*}
	since $F_i\leq f(U)$ and $c(U)-f(U)\sim\sum_{k=i}^{n-1}\Exp(k)$ for all $U\subseteq V$ with $|U|=i$.\\
	Let $Y_i\sim\Exp(1)$ for $i\in[n-1]$ and let $Y_{(i)}$ denote the corresponding order statistics. Then, using R\'enyi's representation (see Corollary~\ref{cor:renyi}), we can derive that
	\begin{equation*}
	\mathbb{P}\left(\sum_{k=i}^{n-1}\Exp(k)<z-F_i\right)=\mathbb{P}\left(Y_{(n-i)}<z-F_i\right).
	\end{equation*}
	Again using the union bound, it follows that
	\begin{align*}
	\mathbb{P}\left(Y_{(n-i)}<z-F_i\right)&=\mathbb{P}\left(\exists J\subseteq[n-1],\,|J|=n-i\,:\,\max_{j\in J}Y_j<z-F_i\right)\\
	&\leq\binom{n-1}{n-i}\cdot\mathbb{P}\left(\max_{j\in[n-i]}Y_j<z-F_i\right)\\
	&=\binom{n-1}{n-i}\left(1-e^{-(z-F_i)}\right)^{n-i}.
	\end{align*}
	By combining the results above, the desired result follows now immediately.\blokje
\end{proof}
Using the result of Lemma~\ref{lemma:expGam} we can also derive a stochastic lower bound for $\mathsf{OPT}_{n-k}$.
\begin{lemma}
	\label{lemma:LBOPTnk}
	Let $Z_k\sim\Gamma(k,e\binom{n}{2}/k)$. Then we have $\mathsf{OPT}_{n-k}\succsim F_{n-k}+Z_k$.
\end{lemma}
\begin{proof}
	If the number of open facilities in a solution is fixed to be $n-k$, then the total opening costs of the optimal solution is trivially lower bounded by $F_{n-k}$. Moreover, the total connection costs in this case is lower bounded by the total length of the $k$ shortest edges in the metric. This in turn can be lower bounded by the total weight of the $k$ lightest edge weights used to generate the metric.
	
	Let $S_k$ denote the sum of the $k$ lightest edge weights. Since all edge weights are independent and standard exponential distributed, we have $S_1\sim\Exp\left(\binom{n}{2}\right)$. Using the memorylessness property of the exponential distribution, it follows that $S_2-S_1\sim S_1+\Exp\left(\binom{n}{2}-1\right)$, i.e., the second lightest edge weight is equal to the lightest edge weight plus the minimum of $\binom{n}{2}-1$ standard exponential distributed random variables. In general, we get $S_{i+1}-S_i\sim S_i+\Exp\left(\binom{n}{2}-i\right)$. This yields
	\begin{equation*}
	S_k\sim\sum_{i=0}^{k-1}(k-i)\cdot\Exp\left(\tbinom{n}{2}-i\right)\sim\sum_{i=0}^{k-1}\Exp\left(\frac{\binom{n}{2}-i}{k-i}\right)\succsim\Gamma\left(k,\frac{e\binom{n}{2}}{k}\right)\sim Z_k,
	\end{equation*}
	where the stochastic dominance follows from Lemma \ref{lemma:expGam} by observing that
	\begin{equation*}
	\prod_{i=0}^{k-1}\frac{\binom{n}{2}-i}{k-i}=\frac{\binom{n}{2}!}{k!\left(\binom{n}{2}-k\right)!}=\binom{\binom{n}{2}}{k}\leq\left(\frac{e\binom{n}{2}}{k}\right)^k,
	\end{equation*}
	where the inequality follows from applying the well-known inequality $\binom{m}{k}\leq(em/k)^k$. The desired result follows now immediately.\blokje
\end{proof}

\section{Main results}\label{sect:probanalysis}

In this section we present our main results. We show that our simple heuristic is asymptotically optimal if $\kappa\in o(n)$ (Theorem \ref{thm:kappasmall}), and we provide a closed-form expression as an upper bound for the expected approximation ratio if $\kappa\in\Theta(n)$ (Theorem \ref{thm:kappalarge}). Finally we will evaluate this expression for the special case where $f_1=\ldots=f_n=f$.
\begin{theorem}
	\label{thm:kappasmall}
	Define $\kappa:=\kappa(n;f_1,\ldots,f_n)=\max\{i\in[n]:f_i<1/(i-1)\}$ and assume that $\kappa\in o(n)$. Let $\mathsf{ALG}$ denote the total cost of the solution which opens, independently of the metric space, the $\kappa$ cheapest facilities (breaking ties arbitrarily), i.e., the facilities with opening costs $f_1,\ldots,f_\kappa$. Then, it follows that
	\begin{equation*}
	\mathbb{E}\left[\frac{\mathsf{ALG}}{\mathsf{OPT}}\right]=1+o(1).
	\end{equation*}
\end{theorem}
In order to prove this theorem, we consider the following three cases for the opening cost $f_1$ of the cheapest facility:
\begin{itemize}
	\item [\textbf{1.}] $f_1\leq1/\ln^2(n)$ as $n\to\infty$;
	\item [\textbf{2.}] $f_1\in O(\ln(n))$ and $f_1>1/\ln^2(n)$ as $n\to\infty$;
	\item [\textbf{3.}] $f_1\in\omega(\ln(n))$.
\end{itemize}
We start with the rather straightforward proof of Case 3.
\begin{proof}[Proof of Theorem \ref{thm:kappasmall} (Case 3)]
	For sufficiently large $n$, we have $f_1>1$, and thus $\kappa=1$ since $f_2\geq f_1>1=1/(2-1)$. Therefore, using our observations in Section~\ref{sect:heur}, we can derive that $\mathbb{E}[\mathsf{ALG}]=F_1+\ln(n)+\Theta(1)$ for sufficiently large $n$. Moreover, we know that $\mathsf{OPT}\geq F_1$. Using this observation, it follows that
	\begin{equation*}
	\mathbb{E}\left[\frac{\mathsf{ALG}}{\mathsf{OPT}}\right]\leq\mathbb{E}\left[\frac{\mathsf{ALG}}{F_1}\right]=\frac{F_1+\ln(n)+\Theta(1)}{F_1}=1+\frac{\ln(n)+\Theta(1)}{\omega(\ln(n))}=1+o(1),
	\end{equation*}
	which finishes the proof of this case.\blokje
\end{proof}
In order to prove Case 1 of Theorem~\ref{thm:kappasmall} we need the following two lemmas.
\begin{lemma}
	\label{lemma:integralbound}
	Let $f_1\leq1/\ln^2(n)$ as $n\to\infty$. For sufficiently large $n$ we have
	\begin{equation*}
	\int_{1/f_1^2}^\infty\mathbb{P}\left(\mathsf{ALG}\geq\sqrt{x}\right)\,\mathrm{d}x\leq O\left(\frac1n\right).
	\end{equation*}
\end{lemma}
\begin{proof}
	We start by providing a bound for the cumulative distribution function of $\mathsf{ALG}$. Let $t\in\mathbb{R}$. By our observations in Section~\ref{sect:heur} we know this distribution, and since $F_\kappa\leq2$, we can bound it as follows
	\begin{equation*}
	\mathbb{P}(\mathsf{ALG}\geq t)=\mathbb{P}\left(\sum_{i=\kappa}^{n-1}\Exp(i)\geq t-F_\kappa\right)\leq\mathbb{P}\left(\sum_{i=\kappa}^{n-1}\Exp(i)\geq t-2\right).
	\end{equation*}
	Now, let $Y_i\sim\Exp(1)$ independently, $i=1,2,\ldots,n-1$, and let $Y_{(i)}$ denote the corresponding order statistics. Using R\'enyi's representation (see Corollary \ref{cor:renyi}), we can now rewrite the last probability as follows:
	\begin{align*}
	\mathbb{P}\left(\sum_{i=\kappa}^{n-1}\Exp(i)\geq t-2\right)&=\mathbb{P}\left(Y_{(n-\kappa)}\geq t-2\right)\\
	&=\mathbb{P}\left(\exists\,L\subseteq[n-1],|L|=\kappa:\min_{j\in L}Y_j\geq t-2\right).
	\end{align*}
	Applying a union bound to this result, we obtain that
	\begin{equation*}
	\mathbb{P}(\mathsf{ALG}\geq t)\leq\binom{n-1}{\kappa}\cdot\mathbb{P}\left(\min_{j\in[\kappa]}Y_j\geq t-2\right)\leq\binom{n-1}{\kappa}\cdot e^{-\kappa(t-2)}.
	\end{equation*}
	Note that the last inequality becomes an equality whenever $t-2\geq0$.\\
	We can use this result to bound the given integral as follows:
	\begin{align*}
	\int_{1/f_1^2}^\infty\mathbb{P}\left(\mathsf{ALG}\geq\sqrt{x}\right)\,\mathrm{d}x&\leq\int_{1/f_1^2}^\infty\binom{n-1}{\kappa}\cdot e^{-\kappa(\sqrt{x}-2)}\,\mathrm{d}x\\
	&=\binom{n-1}{\kappa}e^{2\kappa}\int_{1/f_1^2}^\infty e^{-\kappa\sqrt{x}}\,\mathrm{d}x\\
	&=\binom{n-1}{\kappa}e^{2\kappa}\left(\frac2{\kappa^2}\left(1+\frac{\kappa}{f_1}\right)e^{-\kappa/f_1}\right)\\
	&\leq2n^\kappa\left(1+\frac1{f_1}\right)e^{3\kappa-\kappa/f_1},
	\end{align*}
	where we used $\binom{n-1}{\kappa}\leq(en)^\kappa$ to bound the binomial coefficient. It remains to be shown that $2n^\kappa(1+1/f_1)e^{3\kappa-\kappa/f_1}=O(1/n)$. To do so, we start by claiming that the following inequality holds for sufficiently large $n$:
	\begin{equation*}
	(\kappa+1)\ln(n)+\ln\left(1+\frac1{f_1}\right)+3\kappa\leq\frac{\kappa}{f_1}.
	\end{equation*}
	To see this, observe that for sufficiently large $n$ we have $(\kappa+1)\ln(n)\leq\kappa/3f_1$, $\ln(1+1/f_1)\leq\kappa/3f_1$ and $3\kappa\leq\kappa/3f_1$ (in all three cases since $1/f_1\geq\ln^2(n)$). Rearranging the inequality, we get
	\begin{equation*}
	\kappa\ln(n)+\ln\left(1+\frac1{f_1}\right)+3\kappa-\frac{\kappa}{f_1}\leq-\ln(n).
	\end{equation*}
	Upon exponentiation of both sides we obtain the desired result, which finishes this proof.\blokje
\end{proof}
\begin{lemma}
	\label{lemma:summationterms}
	Let $q$ be a constant such that $f_n/f_1\leq n^q$, let $\beta(n)=\ln(n/\kappa)(1+1/n)^{-1}$, take $\zeta(n):=\max\{i:\beta(n)\geq F_i\}$ and assume that $\kappa<n$. For sufficiently large $n$, and for any integer $i$ with $1\leq i\leq\zeta(n)$, we have
	\begin{equation*}
	\binom{n}{i}\binom{n-1}{i-1}\left(1-e^{-(\beta(n)-F_i)}\right)^{n-i}\leq\frac1{n^{2q+4}}.
	\end{equation*}
\end{lemma}
\begin{proof}
	Let $n$ be sufficiently large. Since $i\ln(en/i)$ is an increasing function of $i$ whenever $0<i<n$, it follows that
	\begin{equation*}
	2i\ln\left(\frac{en}{i}\right)-\left(n-\zeta(n)\right)\cdot e^{F_i-\beta(n)}\leq2\zeta(n)\ln\left(\frac{en}{\zeta(n)}\right)-\left(n-\zeta(n)\right)\cdot e^{-\beta(n)},
	\end{equation*}
	where we also used $e^{F_i}\geq1$ for all $i\in[n]$. Next, define $\alpha(n):=n/\kappa$ and recall that (by construction) $f_{\kappa+c}\geq f_{\kappa+1}\geq1/\kappa$ for all $c\geq1$ and thus $F_{\kappa+c}>c/\kappa$. Using this, we can see that $F_{\kappa+\kappa\beta(n)}>\beta(n)$, from which follows that $\zeta(n)\leq\kappa(1+\beta(n))$ and $\zeta(n)/n\leq(1+\beta(n))/\alpha(n)\leq(1+\ln(\alpha(n)))/\alpha(n)$, where the last inequality follows from the definition of $\beta(n)$. Applying this, we obtain
	\begin{align*}
	&2i\ln\left(\frac{en}{i}\right)-\left(n-\zeta(n)\right)\cdot e^{F_i-\beta(n)}\\
	&\qquad\qquad\qquad\leq n\cdot\left(\frac{2\zeta(n)}{n}\ln\left(\frac{en}{\zeta(n)}\right)-\left(1-\frac{\zeta(n)}{n}\right)\cdot e^{-\beta(n)}\right)\\
	&\qquad\qquad\qquad\leq n\cdot\left(\frac{2+2\ln(\alpha(n))}{\alpha(n)}\ln\left(\frac{e\alpha(n)}{1+\ln(\alpha(n))}\right)-\left(1-\frac{1+\ln(\alpha(n))}{\alpha(n)}\right)\cdot e^{-\beta(n)}\right),
	\end{align*}
	since $0<\zeta(n)/n\leq(1+\ln(\alpha(n)))/\alpha(n)\leq1$.\\
	Since $\kappa\in o(n)$, we have $\alpha(n)\to\infty$ as $n\to\infty$ and $\beta(n)=\ln(\alpha(n))(1+1/n)^{-1}$, implying $e^{-\beta(n)}=\alpha(n)^{1-1/(1+n)}$. This gives us
	\begin{align*}
	&2i\ln\left(\frac{en}i\right)-\left(n-\zeta(n)\right)\cdot e^{F_i-\beta(n)}\\
	&\qquad\qquad\qquad\leq n\cdot\left(\frac{2+2\ln(\alpha(n))}{\alpha(n)}\ln\left(\frac{e\alpha(n)}{1+\ln(\alpha(n))}\right)+\frac{1+\ln(\alpha(n))}{\alpha(n)^{1/(1+n)}}-\alpha(n)^{1-1/(1+n)}\right).
	\end{align*}
	Observe that the dominant term between the brackets on the right-hand side is given by $-\alpha(n)^{1-1/(1+n)}$, implying that this factor becomes less than $-1$ whenever $n$ is sufficiently large. So, we obtain that
	\begin{equation*}
	2i\ln\left(\frac{en}i\right)-\left(n-\zeta(n)\right)\cdot e^{F_i-\beta(n)}\leq-n\leq-(2q+4)\ln(n),
	\end{equation*}
	since $q$ is a constant. Combining this with the well-known inequality $\ln(1-x)\leq-x$ (for $0\leq x<1$), it follows that
	\begin{align*}
	2i\ln\left(\frac{en}i\right)+\left(n-\zeta(n)\right)\ln\left(1-e^{-(\beta(n)-F_i)}\right)&\!\leq\!2i\ln\left(\frac{en}i\right)-\left(n-\zeta(n)\right)\cdot e^{F_i-\beta(n)}\\
	&\!\leq\!-(2q+4)\ln(n).
	\end{align*}
	From this inequality, we immediately get
	\begin{equation*}
	\left(\frac{en}i\right)^{2i}\cdot\left(1-e^{-(\beta(n)-F_i)}\right)^{n-\zeta(n)}\leq\frac1{n^{2q+4}}.
	\end{equation*}
	On the other hand, since $\binom{n-1}{k-1}\leq\binom{n}{k}$, $\binom{n_1}{k_1}\binom{n_2}{k_2}\leq\binom{n_1+n_2}{k_1+k_2}$ and $\binom{n}{k}\leq(en/k)^k$, it follows also that
	\begin{align*}
	\binom{n}{i}\binom{n-1}{i-1}\left(1-e^{-(\beta(n)-F_i)}\right)^{n-i}&\leq\left(\frac{en}i\right)^{2i}\cdot\left(1-e^{-(\beta(n)-F_i)}\right)^{n-i}\\
	&\leq\left(\frac{en}i\right)^{2i}\cdot\left(1-e^{-(\beta(n)-F_i)}\right)^{n-\zeta(n)},
	\end{align*}
	where the last inequality follows since $1-e^{-(\beta(n)-F_i)}\leq1$ and $n-i\geq n-\zeta(n)$.\\ Combining the two results above yields the desired inequality.\blokje
\end{proof}
\begin{proof}[Proof of Theorem \ref{thm:kappasmall} (Case 1)]
	Let $n$ be sufficiently large. By definition of $\kappa$, it follows that $f_\kappa<1/(\kappa-1)$ and thus $F_\kappa<\kappa/(\kappa-1)\leq2$ whenever $\kappa\geq2$. If $\kappa=1$, then we have $F_\kappa=f_1<1$ as $n\to\infty$. So, in any case we have $F_\kappa=O(1)$. Now, by our observations in Section~\ref{sect:heur} we know that $\mathbb{E}[\mathsf{ALG}]=F_\kappa+\ln(n/\kappa)+\Theta(1)=\ln(n/\kappa)+\Theta(1)$. Set $\beta(n):=\ln(n/\kappa)(1+1/n)^{-1}$ and observe that $\beta(n)\in\omega(1)$.\\
	Conditioning on the events $\mathsf{OPT}\geq\beta(n)$ and $\mathsf{OPT}<\beta(n)$ yields
	\begin{equation*}
	\mathbb{E}\left[\frac{\mathsf{ALG}}{\mathsf{OPT}}\right]\leq\mathbb{E}\left[\frac{\mathsf{ALG}}{\beta(n)}\right]+\mathbb{P}\left(\mathsf{OPT}<\beta(n)\right)\cdot\mathbb{E}\left[\frac{\mathsf{ALG}}{\mathsf{OPT}}\;\middle|\;\mathsf{OPT}<\beta(n)\right].
	\end{equation*}
	We start by bounding the second part. Applying Lemma \ref{lemma:E[X/Y]}, with $X=\mathsf{ALG}$, $Y=\mathsf{OPT}$, $y=\beta(n)$ and $\delta=f_1$, we get
	\begin{equation*}
	\mathbb{P}\left(\mathsf{OPT}<\beta(n)\right)\mathbb{E}\left[\frac{\mathsf{ALG}}{\mathsf{OPT}}\;\middle|\;\mathsf{OPT}<\beta(n)\right]\leq\frac{\mathbb{P}\left(\mathsf{OPT}<\beta(n)\right)}{f_1^2}+\int\limits_{1/f_1^2}^\infty\mathbb{P}\left(\mathsf{ALG}\geq\sqrt{x}\right)\,\mathrm{d}x.
	\end{equation*}
	Note that we may use Lemma \ref{lemma:E[X/Y]} since $\mathsf{OPT}\geq f_1$ and $\beta(n)>f_1$, which implies $\mathbb{P}(\mathsf{OPT}<\beta(n))>0$. The probability containing $\mathsf{OPT}$ can be bounded using Lemma~\ref{lemma:OPTsmall}, whereas the integral can be bounded by Lemma~\ref{lemma:integralbound}.	Together, this yields
	\begin{equation*}
	\mathbb{E}\left[\frac{\mathsf{ALG}}{\mathsf{OPT}}\right]\leq\mathbb{E}\left[\frac{\mathsf{ALG}}{\beta(n)}\right]+\frac1{f_1^2}\cdot\sum_{i=1}^{\zeta(n)}\binom{n}{i}\binom{n-1}{i-1}\left(1-e^{-(\beta(n)-F_i)}\right)^{n-i}+O\left(\frac1n\right),
	\end{equation*}
	where $\zeta(n):=\max\{i:\beta(n)\geq F_i\}$. The terms of the summation can be bounded by Lemma~\ref{lemma:summationterms}. Using this lemma, we obtain that
	\begin{equation*}
	\mathbb{E}\left[\frac{\mathsf{ALG}}{\mathsf{OPT}}\right]\leq\mathbb{E}\left[\frac{\mathsf{ALG}}{\beta(n)}\right]+\frac1{f_1^2}\cdot\sum_{i=1}^{\zeta(n)}\frac1{n^{2q+4}}+O\left(\frac1n\right)\leq\mathbb{E}\left[\frac{\mathsf{ALG}}{\beta(n)}\right]+\frac{1/f_1^2}{n^{2q+3}}+O\left(\frac1n\right),
	\end{equation*}
	since $\zeta(n)\leq n$ by definition. Moreover, since $\kappa\in o(n)$ implies $f_n>1/n$ as $n\to\infty$, we also have $f_1\geq f_n/n^q>1/n^{q+1}$ as $n\to\infty$ for some constant $q$. This results in
	\begin{equation*}
	\mathbb{E}\left[\frac{\mathsf{ALG}}{\mathsf{OPT}}\right]\leq\mathbb{E}\left[\frac{\mathsf{ALG}}{\beta(n)}\right]+n^{2q+2}\cdot\frac1{n^{2q+3}}+O\left(\frac1n\right)=\mathbb{E}\left[\frac{\mathsf{ALG}}{\beta(n)}\right]+O\left(\frac1n\right).
	\end{equation*}
	Since we started with $\beta(n)=\ln(n/\kappa)(1+1/n)^{-1}$ and $n/\kappa\in\omega(1)$ (since $\kappa\in o(n)$), it follows that
	\begin{equation*}
	\mathbb{E}\left[\frac{\mathsf{ALG}}{\mathsf{OPT}}\right]\leq\frac{\mathbb{E}[\mathsf{ALG}]}{\beta(n)}+O\left(\frac1n\right)\leq\frac{\ln(n/\kappa)+\Theta(1)}{\ln(n/\kappa)}\left(1+\frac1{n}\right)+O\left(\frac1n\right)=1+o(1),
	\end{equation*}
	which finishes the proof of this case.\blokje
\end{proof}
In order to prove Case 2 of Theorem~\ref{thm:kappasmall} we need the following five lemmas.
\begin{lemma}
	\label{lemma:probAandB}
	Let $A$ and $B$ be two arbitrary events. Then we have $\mathbb{P}(A\cap B)\leq2\sqrt{\mathbb{P}(A)\mathbb{P}(B)}$.
\end{lemma}
\begin{proof}
	Let $X$ and $Y$ denote the indicator variables corresponding to $A$ and $B$, respectively. Then it follows that $\mathbb{E}[X]=\mathbb{E}[X^2]=\mathbb{P}(A)$ and $\mathbb{E}[Y]=\mathbb{E}[Y^2]=\mathbb{P}(B)$. From this, we deduce that $\Var(X)\leq\mathbb{P}(A)$ and $\Var(Y)\leq\mathbb{P}(B)$. Moreover, we can see that $\mathbb{E}[XY]=\mathbb{P}(A\cap B)$. Now, combining this knowledge with the variance-bound for the covariance, we derive
	\begin{align*}
	\mathbb{P}(A\cap B)&=\mathbb{E}[XY]=\mathbb{E}[X]\mathbb{E}[Y]+\Cov(X,Y)\leq\mathbb{P}(A)\mathbb{P}(B)+\sqrt{\Var(X)\Var(Y)}\\
	&\leq\mathbb{P}(A)\mathbb{P}(B)+\sqrt{\mathbb{P}(A)\mathbb{P}(B)}.
	\end{align*}
	Since $0\leq\mathbb{P}(A)\mathbb{P}(B)\leq1$, it follows that $\mathbb{P}(A)\mathbb{P}(B)\leq\sqrt{\mathbb{P}(A)\mathbb{P}(B)}$, which finishes this proof.\blokje
\end{proof}
\begin{lemma}[{\cite[Lemma~3.2]{Bringmann2015}}]
	\label{lemma:hat-c}
	Let $X\sim\sum_{i=1}^n\Exp(ci)$. Then $\mathbb{P}(X\leq\alpha)=(1-e^{-c\alpha})^n$ for any $\alpha\geq0$.
\end{lemma}
\begin{lemma}
	\label{lemma:c(1)integral}
	Suppose that $F_\kappa\in O(\ln(n))$. Set $m:=2n-1$ to shorten notation. Then, for sufficiently large $n$ we have
	\begin{equation*}
	F_\kappa+\int_{F_\kappa}^\infty\sqrt{1-\left(1-e^{-(x-F_\kappa)}\right)^{n-1}}\,\mathrm{d}x\leq2\sqrt{em}.
	\end{equation*}
\end{lemma}
\begin{proof}
	Let $n$ be sufficiently large. We start by applying the change of variables $y=x-F_\kappa$ to the integral. This yields
	\begin{equation*}
	F_\kappa+\int_{F_\kappa}^\infty\sqrt{1-\left(1-e^{-(x-F_\kappa)}\right)^{n-1}}\,\mathrm{d}x=F_\kappa+\int_0^\infty\sqrt{1-\left(1-e^{-y}\right)^{n-1}}\,\mathrm{d}y.
	\end{equation*}
	Next, we use Bernoulli's inequality to obtain that
	\begin{align*}
	F_\kappa+\int_0^\infty\sqrt{1-\left(1-e^{-y}\right)^{n-1}}\,\mathrm{d}y&\leq F_\kappa+\int_0^\infty\sqrt{1-\left(1-(n-1)e^{-y}\right)}\,\mathrm{d}y\\
	&=F_\kappa+\sqrt{n-1}\cdot\int_0^\infty e^{-\frac12y}\,\mathrm{d}y=F_\kappa+2\sqrt{n-1}.
	\end{align*}
	Finally, since $F_\kappa\in O(\ln(n))$, we have $F_\kappa\leq\sqrt{n-1}$ for $n$ sufficiently large. Combining this with the inequality $3<2\sqrt{e}$, it follows that $F_\kappa+2\sqrt{n-1}\leq3\sqrt{n-1}\leq3\sqrt{2n-1}=3\sqrt{m}\leq2\sqrt{em}$, which finishes this proof.\blokje
\end{proof}
\begin{lemma}
	\label{lemma:sumbound}
	Set $m:=2n-1$ to shorten notation. Let $f_1,\beta(n)>0$ and $0<\zeta(n)\leq n$. Then, for sufficiently large $n$, we have
	\begin{align*}
	&\frac{4\sqrt{em}}{f_1}\cdot\sqrt{\sum_{i=1}^{\zeta(n)}\binom{n}{i}\binom{n-1}{i-1}\left(1-e^{-(\beta(n)-F_i)}\right)^{n-i}}\\
	&\qquad\qquad\qquad\leq\frac5{f_1}\cdot\left(em\right)^{\zeta(n)}\cdot e^{\beta(n)}\cdot\left(1-e^{-(\beta(n)-f_1)}\right)^{\frac12n-\frac12\zeta(n)}.
	\end{align*}
\end{lemma}
\begin{proof}
	We start by bounding the product of the binomials using the inequalities $\binom{n_1}{k_1}\binom{n_2}{k_2}\leq\binom{n_1+n_2}{k_1+k_2}$ and $\binom{n}{k}\leq(en/k)^k\leq(en)^k$ for $1\leq k\leq n$. This results in
	\begin{equation*}
	\sum_{i=1}^{\zeta(n)}\binom{n}{i}\binom{n-1}{i-1}\left(1-e^{-(\beta(n)-F_i)}\right)^{n-i}\leq\sum_{i=1}^{\zeta(n)}(em)^{2i-1}\left(1-e^{-(\beta(n)-F_i)}\right)^{n-i}.
	\end{equation*}
	Next we use the Cauchy-Schwarz inequality for summations to obtain that
	\begin{align*}
	\sum_{i=1}^{\zeta(n)}(em)^{2i-1}\left(1-e^{-(\beta(n)-F_i)}\right)^{n-i}&\leq\sqrt{\sum_{i=1}^{\zeta(n)}(em)^{4i-2}\cdot\sum_{i=1}^{\zeta(n)}\left(1-e^{-(\beta(n)-F_i)}\right)^{2n-2i}}\\
	&\leq\sqrt{\sum_{i=1}^{\zeta(n)}(em)^{4i-2}\cdot\sum_{i=1}^{\zeta(n)}\left(1-e^{-(\beta(n)-f_1)}\right)^{2n-2i}},
	\end{align*}
	where we also applied the inequality $F_i\geq f_1$. Now we can compute and bound both summations, which yields
	\begin{equation*}
	\sum_{i=1}^{\zeta(n)}(em)^{4i-2}=(em)^2\cdot\frac{(em)^{4\zeta(n)}-1}{(em)^4-1}\leq(em)^2\cdot\frac{(em)^{4\zeta(n)}}{(em)^4-1},
	\end{equation*}
	and
	\begin{align*}
	\sum_{i=1}^{\zeta(n)}\left(1-e^{-(\beta(n)-f_1)}\right)^{2n-2i}&=\left(1-e^{-(\beta(n)-f_1)}\right)^{2n-2\zeta(n)}\cdot\frac{1-\left(1-e^{-(\beta(n)-f_1)}\right)^{2\zeta(n)}}{1-\left(1-e^{-(\beta(n)-f_1)}\right)^2}\\
	&\leq\frac{\left(1-e^{-(\beta(n)-f_1)}\right)^{2n-2\zeta(n)}}{1-\left(1-e^{-(\beta(n)-f_1)}\right)^2}\leq\frac{\left(1-e^{-(\beta(n)-f_1)}\right)^{2n-2\zeta(n)}}{e^{-(\beta(n)-f_1)}}.
	\end{align*}
	Combining the results above, and multiplying them with $4\sqrt{em}/f_1$, we get
	\begin{align*}
	&\frac{4\sqrt{em}}{f_1}\cdot\sqrt{\sum_{i=1}^{\zeta(n)}\binom{n}{i}\binom{n-1}{i-1}\left(1-e^{-(\beta(n)-F_i)}\right)^{n-i}}\\
	&\qquad\qquad\qquad\leq\frac{4\sqrt{em}}{f_1}\cdot\sqrt[4]{(em)^2\cdot\frac{(em)^{4\zeta(n)}}{(em)^4-1}\cdot\frac{\left(1-e^{-(\beta(n)-f_1)}\right)^{2n-2\zeta(n)}}{e^{-(\beta(n)-f_1)}}}.
	\end{align*}
	Next, upon rewriting and applying the inequalities $x/(x-1)\leq2$ (for $x\geq2$) and $4\cdot\sqrt[4]{2}<5$ we obtain
	\begin{align*}
	&\frac{4\sqrt{em}}{f_1}\cdot\sqrt[4]{(em)^2\cdot\frac{(em)^{4\zeta(n)}}{(em)^4-1}\cdot\frac{\left(1-e^{-(\beta(n)-f_1)}\right)^{2n-2\zeta(n)}}{e^{-(\beta(n)-f_1)}}}\\
	&\qquad\qquad\qquad=\frac{4}{f_1}\cdot\sqrt[4]{\frac{(em)^4}{(em)^4-1}\cdot(em)^{4\zeta(n)}\cdot\frac{\left(1-e^{-(\beta(n)-f_1)}\right)^{2n-2\zeta(n)}}{e^{-(\beta(n)-f_1)}}}\\
	&\qquad\qquad\qquad\leq\frac{5}{f_1}\cdot\sqrt[4]{(em)^{4\zeta(n)}\cdot\frac{\left(1-e^{-(\beta(n)-f_1)}\right)^{2n-2\zeta(n)}}{e^{-(\beta(n)-f_1)}}}.
	\end{align*}
	Applying the inequalities $e^{-f_1}\leq1$ and $e^{\beta(n)/4}\leq e^{\beta(n)}$ yields the desired result.\blokje
\end{proof}
\begin{lemma}
	\label{lemma:orderbound}
	Set $m:=2n-1$ to shorten notation. Let $\beta(n)=\ln(n)(\phi(n)+1)(1+1/\ln(\ln(n)))^{-1}$ where $\phi(n)=f_1/\ln(n)$, and let $\zeta(n)=\max\{i:\beta(n)\geq F_i\}$. Suppose that $f_1>1/ln^2(n)$ as $n\to\infty$. For sufficiently large $n$, we have
	\begin{align*}
	\frac5{f_1}\cdot\left(em\right)^{\zeta(n)}\cdot e^{\beta(n)}\cdot\left(1-e^{-(\beta(n)-f_1)}\right)^{\frac12n-\frac12\zeta(n)}\leq O\left(\frac1n\right).
	\end{align*}
\end{lemma}
\begin{proof}
	First observe that, by definition of $\zeta(n)$, we have $\zeta(n)\leq\beta(n)/f_1$, since $F_k\geq kf_1$ for all $k\in[n]$. We will now show that the following inequality holds for sufficiently large $n$:
	\begin{equation*}
	\zeta(n)\ln(em)+\beta(n)\leq-2\ln(n)+\left(\tfrac12n-\tfrac12\zeta(n)\right)\cdot e^{f_1}\cdot e^{-\beta(n)}.
	\end{equation*}
	When analyzing the left-hand side, we can see that $\zeta(n)\leq\beta(n)/f_1=O(\ln(n))\cdot O(\ln^2(n))=O(\ln^3(n))$ and $\beta(n)=O(\ln(n))$. So, the left-hand side is bounded by $O(\ln^4(n))$. When analyzing the right-hand side, we can see that $\frac12n-\frac12\zeta(n)=\Omega(n)$, since $\zeta(n)\leq O(\ln^3(n))$. Moreover, we have $e^{f_1}=n^{\phi(n)}$ and $e^{-\beta(n)}=n^{(\phi(n)+1)(-1+1/(1+\ln(\ln(n))))}$, where we used the equality $1/(1+1/x)=x/(1+x)=1-1/(1+x)$ for $x=\ln(\ln(n))$. From this we can deduce that $e^{f_1}\cdot e^{-\beta(n)}\geq\Omega(n^{-1+1/(1+\ln(\ln(n)))})$. So, the right-hand side is bounded by $\Omega(n^{1/(1+\ln(\ln(n)))})$.\\
	Since $O(\ln^4(n))<\Omega(n^{1/(1+\ln(\ln(n)))})$, the stated inequality follows. Now, rewriting the right-hand side and then applying the well-known inequality $1-x\leq-\ln(x)$ for $x\geq0$ yields
	\begin{equation*}
	\zeta(n)\ln(em)+\beta(n)\leq-2\ln(n)+\left(\tfrac12\zeta(n)-\tfrac12n\right)\cdot\ln\left(1-e^{-(\beta(n)-f_1)}\right).
	\end{equation*}
	From this inequality, we immediately get
	\begin{equation*}
	(em)^{\zeta(n)}\cdot e^{\beta(n)}\leq\frac1{n^2}\cdot\left(1-e^{-(\beta(n)-f_1)}\right)^{\frac12\zeta(n)-\frac12n}.
	\end{equation*}
	Using this inequality, in combination with $1/{f_1}\leq\ln^2(n)$ and $\ln^2(n)\cdot O(1/n^2)\leq O(1/n)$, the desired result follows.\blokje
\end{proof}
\begin{proof}[Proof of Theorem \ref{thm:kappasmall} (Case 2)]
	Recall that in Case 2 we have $f_1\in O(\ln(n))$ and $f_1>1/\ln^2(n)$ as $n\to\infty$.
	Let $n$ be sufficiently large. By definition of $\kappa$, it follows that $f_\kappa<1/(\kappa-1)$ and thus $F_\kappa<\kappa/(\kappa-1)\leq2$ whenever $\kappa\geq2$. If $\kappa=1$, then we have $F_\kappa=f_1=O(\ln(n))$. Now, by our observations in Section~\ref{sect:heur} we know that $\mathbb{E}[\mathsf{ALG}]=F_\kappa+\ln(n/\kappa)+\Theta(1)\leq F_\kappa+\ln(n)+\Theta(1)$.
	Now, set $\beta(n):=\ln(n)(\phi(n)+1)(1+1/\ln(\ln(n)))^{-1}$, where $\phi(n):=f_1/\ln(n)$. Observe that $\beta(n)>F_\kappa$ for sufficiently large $n$.\\
	Conditioning on the events $\mathsf{OPT}\geq\beta(n)$ and $\mathsf{OPT}<\beta(n)$ yields
	\begin{equation*}
	\mathbb{E}\left[\frac{\mathsf{ALG}}{\mathsf{OPT}}\right]\leq\mathbb{E}\left[\frac{\mathsf{ALG}}{\beta(n)}\right]+\mathbb{P}\left(\mathsf{OPT}<\beta(n)\right)\cdot\mathbb{E}\left[\frac{\mathsf{ALG}}{\mathsf{OPT}}\;\middle|\;\mathsf{OPT}<\beta(n)\right].
	\end{equation*}
	We start by bounding the second part. Since $\mathsf{OPT}\geq f_1$ by definition, we may bound and subsequently rewrite the second part as follows:
	\begin{align*}
	\mathbb{P}\left(\mathsf{OPT}<\beta(n)\right)\cdot\mathbb{E}\left[\frac{\mathsf{ALG}}{\mathsf{OPT}}\;\middle|\;\mathsf{OPT}<\beta(n)\right]&\leq\frac1{f_1}\cdot\mathbb{P}\left(\mathsf{OPT}<\beta(n)\right)\cdot\mathbb{E}\left[\mathsf{ALG}\;\middle|\;\mathsf{OPT}<\beta(n)\right]\\
	&=\frac1{f_1}\cdot\int_0^\infty\mathbb{P}\left(\mathsf{ALG}>x\text{ and }\mathsf{OPT}<\beta(n)\right)\,\mathrm{d}x.
	\end{align*}
	Since the events $\mathsf{ALG}>x$ and $\mathsf{OPT}<\beta(n)$ are dependent, we use Lemma \ref{lemma:probAandB} to bound the probability inside the integral. This results in the following:
	\begin{equation*}
	\mathbb{P}\left(\mathsf{OPT}<\beta(n)\right)\cdot\mathbb{E}\left[\frac{\mathsf{ALG}}{\mathsf{OPT}}\;\middle|\;\mathsf{OPT}<\beta(n)\right]\leq\frac2{f_1}\cdot\sqrt{\mathbb{P}\left(\mathsf{OPT}<\beta(n)\right)}\cdot\int_0^\infty\sqrt{\mathbb{P}\left(\mathsf{ALG}>x\right)}\,\mathrm{d}x.
	\end{equation*}
	Now recall from Section~\ref{sect:heur} that we know the probability distribution of $\mathsf{ALG}$, namely $\mathsf{ALG}\sim F_\kappa+\sum_{i=\kappa}^{n-1}\Exp(i)$. Now define $\mathsf{UB}\sim F_\kappa+\sum_{i=1}^{n-1}\Exp(i)$, and observe that $\mathsf{ALG}\precsim\mathsf{UB}$. So, it follows that $\mathbb{P}(\mathsf{ALG}>x)\leq\mathbb{P}(\mathsf{UB}>x)$ for all $x\in\mathbb{R}$. Moreover, using Lemma \ref{lemma:hat-c} we can see that $\mathbb{P}(\mathsf{UB}>x)=1$ for $x<F_\kappa$ and $\mathbb{P}(\mathsf{UB}>x)=1-(1-e^{-(x-F_\kappa)})^{n-1}$ for $x\geq F_\kappa$. Using this, it follows that
	\begin{align*}
	\int_0^\infty\sqrt{\mathbb{P}\left(\mathsf{ALG}>x\right)}\,\mathrm{d}x&\leq\int_0^\infty\sqrt{\mathbb{P}\left(\mathsf{UB}>x\right)}\,\mathrm{d}x\\
	&\leq F_\kappa+\int_{F_\kappa}^\infty\sqrt{1-\left(1-e^{-(x-F_\kappa)}\right)^{n-1}}\,\mathrm{d}x.
	\end{align*}
	We will use Lemma~\ref{lemma:c(1)integral} to bound this last expression. Using Lemma~\ref{lemma:OPTsmall}, we can bound the probability involving $\mathsf{OPT}$ as follows, where $\zeta(n):=\max\{i:\beta(n)\geq F_i\}$:
	\begin{equation*}
	\mathbb{P}\left(\mathsf{OPT}<\beta(n)\right)\leq\sum_{i=1}^{\zeta(n)}\binom{n}{i}\binom{n-1}{i-1}\left(1-e^{-(\beta(n)-F_i)}\right)^{n-i}.
	\end{equation*}
	Combining all this information with Lemma~\ref{lemma:sumbound}, we obtain that
	\begin{equation*}
	\mathbb{P}\left(\mathsf{OPT}<\beta(n)\right)\cdot\mathbb{E}\left[\frac{c(1)}{\mathsf{OPT}}\;\middle|\;\mathsf{OPT}<\beta(n)\right]\leq\frac5{f_1}\cdot\left(em\right)^{\zeta(n)}\cdot e^{\beta(n)}\cdot\left(1-e^{-(\beta(n)-f_1)}\right)^{\frac12n-\frac12\zeta(n)}.
	\end{equation*}
	Lemma~\ref{lemma:orderbound} shows that we can bound this expression by $O(1/n)$. Going back to our initial expected approximation ratio, we can now derive that
	\begin{align*}
	\mathbb{E}\left[\frac{\mathsf{ALG}}{\mathsf{OPT}}\right]&\leq\frac{\mathbb{E}[\mathsf{ALG}]}{\beta(n)}+O\left(\frac1n\right)\\
	&\leq\frac{F_\kappa+\ln(n)+\Theta(1)}{\beta(n)}+O\left(\frac1n\right)=\frac{F_\kappa+\ln(n)}{\beta(n)}+O\left(\frac1{\ln(n)}\right).
	\end{align*}
	Now, if $f_1\in o(\ln(n))$ we have $\phi(n)=o(1)$ and $F_\kappa\in o(\ln(n))$, so it follows that
	\begin{align*}
	\mathbb{E}\left[\frac{\mathsf{ALG}}{\mathsf{OPT}}\right]&\leq\frac{F_\kappa+\ln(n)}{\beta(n)}+O\left(\frac1{\ln(n)}\right)\\
	&=\frac{(o(1)+1)}{(\phi(n)+1)}\left(1+\frac1{\ln(\ln(n))}\right)+O\left(\frac1{\ln(n)}\right)=1+o(1),
	\end{align*}
	and if $f_1\in\Theta(\ln(n))$ we have $F_\kappa=f_1=\phi(n)\ln(n)$, so it follows that
	\begin{align*}
	\mathbb{E}\left[\frac{\mathsf{ALG}}{\mathsf{OPT}}\right]&\leq\frac{F_\kappa+\ln(n)}{\beta(n)}+o(1)\\
	&=\frac{(\phi(n)+1)\ln(n)}{\beta(n)}+o(1)=1+\frac1{\ln(\ln(n))}+o(1)=1+o(1),
	\end{align*}
	which finishes the proof of this case.\blokje
\end{proof}
\begin{theorem}
	\label{thm:kappalarge}
	Define $\kappa:=\kappa(n;f_1,\ldots,f_n)=\max\{i\in[n]:f_i<1/(i-1)\}$ and assume that $\kappa\in\Theta(n)$. Let $\mathsf{ALG}$ denote the total cost of the solution which opens, independently of the metric space, the $\kappa$ cheapest facilities (breaking ties arbitrarily), i.e., the facilities with opening costs $f_1,\ldots,f_\kappa$. Then we can bound the expected approximation ratio by
	\begin{align*}
	&\mathbb{E}\left[\frac{\mathsf{ALG}}{\mathsf{OPT}}\right]\leq\sqrt{\max\left\{\frac1{F_n^2},\max_{k\in[n-1]}O\left(\frac{n^4F_{n-k}+kn^2}{n^4F_{n-k}^3+k^4F_{n-k}}\right)\right\}}\\
	&\qquad\qquad\qquad+\sqrt[4]{O\left(\frac{n^6F_{n-1}^2+n^2}{n^8F_{n-1}^7+F_{n-1}^3}+\frac{n^{10}F_{n-2}^3+n^4}{n^{14}F_{n-2}^9+F_{n-2}^2}+\sum_{k=3}^{n-1}\frac{k^3n^{12}F_{n-k}^3+k^9n^6}{n^{16}F_{n-k}^9+k^{16}F_{n-k}}\right)}.
	\end{align*}
	Moreover, if $\kappa=n$, then the expected approximation ratio can be bounded by
	\begin{align*}
	&\mathbb{E}\left[\frac{\mathsf{ALG}}{\mathsf{OPT}}\right]\leq F_n\cdot\sqrt{\max\left\{\frac1{F_n^2},\max_{k\in[n-1]}O\left(\frac{n^4F_{n-k}+kn^2}{n^4F_{n-k}^3+k^4F_{n-k}}\right)\right\}}\\
	&\qquad\qquad\qquad+F_n\cdot\sqrt[4]{O\left(\frac{n^6F_{n-1}^2+n^2}{n^8F_{n-1}^7+F_{n-1}^3}+\frac{n^{10}F_{n-2}^3+n^4}{n^{14}F_{n-2}^9+F_{n-2}^2}+\sum_{k=3}^{n-1}\frac{k^3n^{12}F_{n-k}^3+k^9n^6}{n^{16}F_{n-k}^9+k^{16}F_{n-k}}\right)}.
	\end{align*}
\end{theorem}
The proof of this theorem requires some tedious computations which we used to bound exponential integrals by Pad\'e approximants \cite{Luke1969}. The results of these computations are stated in the following lemmas.
\begin{lemma}
	\label{lemma:EXk}
	Let $X_k=1/(F_{n-k}+Z_k)^2$ where $Z_k\sim\Gamma(k,e\binom{n}{2}/k)$. Then, for any $k\in\{1,\ldots,n-1\}$ it follows that
	\begin{equation*}
	\mathbb{E}\left[X_k\right]\leq O\left(\frac{n^4F_{n-k}+kn^2}{n^4F_{n-k}^3+k^4F_{n-k}}\right).
	\end{equation*}
\end{lemma}
\begin{proof}
	Let $z_k(x)$ denote the density function of $Z_k$, i.e.,
	\begin{equation*}
	z_k(x)=\binom{e\binom{n}{2}}{k}^k\cdot\frac{x^{k-1}e^{-ex\binom{n}{2}/k}}{(k-1)!}.
	\end{equation*}		
	For $k=1$ we obtain upon direct computation that
	\begin{align*}
	\mathbb{E}\left[X_1\right]&=\int_0^\infty\frac{z_1(x)}{\left(F_{n-1}+x\right)^2}\,\mathrm{d}x=e\tbinom{n}{2}\cdot\int_0^\infty\frac{e^{-ex\binom{n}{2}}}{\left(F_{n-1}+x\right)^2}\,\mathrm{d}x\\
	&=e\tbinom{n}{2}\cdot\left(\frac1{F_{n-1}}-\frac1{F_{n-1}}\cdot e\tbinom{n}{2}F_{n-1}\cdot e^{e\binom{n}{2}F_{n-1}}\cdot\int\limits_{e\binom{n}{2}F_{n-1}}^\infty\frac{e^{-t}}{t}\,\mathrm{d}t\right).
	\end{align*}
	The second Pad\'e approximant for the exponential integral in this expression \cite{Luke1969} states that for any $\alpha>0$ we have
	\begin{equation*}
	\frac{\alpha^2+3\alpha}{\alpha^2+4\alpha+2}\leq\alpha e^\alpha\int_\alpha^\infty\frac{e^{-t}}{t}\,\mathrm{d}t\leq\frac{\alpha^2+5\alpha+2}{\alpha^2+6\alpha+6}.
	\end{equation*}
	Applying this inequality yields
	\begin{align*}
	\mathbb{E}\left[X_1\right]&\leq e\tbinom{n}{2}\cdot\left(\frac1{F_{n-1}}-\frac1{F_{n-1}}\cdot\frac{e^2\binom{n}{2}^2F_{n-1}^2+3e\binom{n}{2}F_{n-1}}{e^2\binom{n}{2}^2F_{n-1}^2+4e\binom{n}{2}F_{n-1}+2}\right)\\
	&=\frac{e^2\binom{n}{2}^2F_{n-1}+2e\binom{n}{2}}{e^2\binom{n}{2}^2F_{n-1}^3+4e\binom{n}{2}F_{n-1}^2+2F_{n-1}}=O\left(\frac{n^4F_{n-1}+n^2}{n^4F_{n-1}^3+F_{n-1}}\right),
	\end{align*}
	which satisfies the given bound.\\
	For $k>1$ we obtain upon direct computation that
	\begin{align*}
	\mathbb{E}\left[X_k\right]&=\int_0^\infty\frac{z_k(x)}{\left(F_{n-1}+x\right)^2}\,\mathrm{d}x=\frac{e^k\binom{n}{2}^k}{k^k(k-1)!}\cdot\int_0^\infty\frac{x^{k-1}e^{-ex\binom{n}{2}/k}}{\left(F_{n-k}+x\right)^2}\,\mathrm{d}x\\
	&=\frac{e^2\binom{n}{2}^2}{k^2(k-1)}\cdot\left(-1+\frac{k-1+\alpha}{\alpha}\cdot\alpha^{k-1}e^\alpha\int_\alpha^\infty\frac{e^{-t}}{t^{k-1}},\mathrm{d}t\right),
	\end{align*}
	where we used $\alpha:=e\binom{n}{2}F_{n-k}/k$ to shorten notation. The second Pad\'e approximant for the generalized exponential integral in this expression \cite{Luke1969} states that for any $\alpha>0$ and $k>1$ we have
	\begin{equation*}
	\frac{\alpha^2+(k+1)\alpha}{\alpha^2+2k\alpha+k(k-1)}\leq\alpha^{k-1}e^\alpha\int_\alpha^\infty\frac{e^{-t}}{t^{k-1}}\,\mathrm{d}t\leq\frac{\alpha^2+(k+3)\alpha+2}{\alpha^2+2(k+1)\alpha+k(k+1)}.
	\end{equation*}
	Applying this inequality yields
	\begin{align*}
	\mathbb{E}\left[X_k\right]&\leq\frac{e^2\binom{n}{2}^2}{k^2(k-1)}\cdot\left(-1+\frac{k-1+\alpha}{\alpha}\cdot\frac{\alpha^2+(k+3)\alpha+2}{\alpha^2+2(k+1)\alpha+k(k+1)}\right)\\
	&=\frac{e^2\binom{n}{2}^2F_{n-k}+2ke\binom{n}{2}}{e^2\binom{n}{2}^2F_{n-k}^3+2k(k+1)e\binom{n}{2}F_{n-k}^2+k^3(k+1)F_{n-k}}\\
	&=O\left(\frac{n^4F_{n-k}+kn^2}{n^4F_{n-k}^3+k^4F_{n-k}}\right),
	\end{align*}
	which completes this proof.\blokje
\end{proof}
\begin{lemma}
	\label{lemma:VarXk}
	Let $X_k=1/(F_{n-k}+Z_k)^2$ where $Z_k\sim\Gamma(k,e\binom{n}{2}/k)$. Then, for any $k\in\{3,\ldots,n-1\}$ it follows that
	\begin{equation*}
	\mathbb{E}\left[X_k^2\right]-\left(\mathbb{E}\left[X_k\right]\right)^2\leq O\left(\frac{k^3n^{12}F_{n-k}^3+k^9n^6}{n^{16}F_{n-k}^9+k^{16}F_{n-k}}\right),
	\end{equation*}
	whereas for $k=1$ and $k=2$ we have
	\begin{align*}
	\mathbb{E}\left[X_1^2\right]-\left(\mathbb{E}\left[X_1\right]\right)^2&\leq O\left(\frac{n^6F_{n-1}^2+n^2}{n^8F_{n-1}^7+F_{n-1}^3}\right),\\
	\mathbb{E}\left[X_2^2\right]-\left(\mathbb{E}\left[X_2\right]\right)^2&\leq O\left(\frac{n^{10}F_{n-2}^3+n^4}{n^{14}F_{n-2}^9+F_{n-2}^2}\right).
	\end{align*}
\end{lemma}
\begin{proof}
	Let $z_k(x)$ denote the density function of $Z_k$, i.e.,
	\begin{equation*}
	z_k(x)=\binom{e\binom{n}{2}}{k}^k\cdot\frac{x^{k-1}e^{-ex\binom{n}{2}/k}}{(k-1)!}.
	\end{equation*}
	We start by providing lower bounds for $\mathbb{E}[X_k]$. For $k=1$ we can combine our observations in the proof of Lemma~\ref{lemma:EXk} with the following first Pad\'e approximant for the exponential integral \cite{Luke1969}:
	\begin{equation*}
	\frac{\alpha}{\alpha+1}\leq\alpha e^\alpha\int_\alpha^\infty\frac{e^{-t}}{t}\,\mathrm{d}t\leq\frac{\alpha+1}{\alpha+2}.
	\end{equation*}
	This yields
	\begin{align*}
	\mathbb{E}\left[X_1\right]&=e\tbinom{n}{2}\cdot\left(\frac1{F_{n-1}}-\frac1{F_{n-1}}\cdot e\tbinom{n}{2}F_{n-1}\cdot e^{e\binom{n}{2}F_{n-1}}\cdot\int\limits_{e\binom{n}{2}F_{n-1}}^\infty\frac{e^{-t}}{t}\,\mathrm{d}t\right)\\
	&\geq e\tbinom{n}{2}\cdot\left(\frac1{F_{n-1}}-\frac1{F_{n-1}}\cdot\frac{e\binom{n}{2}F_{n-1}+1}{e\binom{n}{2}F_{n-1}+2}\right)\\
	&=\frac{e\binom{n}{2}}{e\binom{n}{2}F_{n-1}^2+2F_{n-1}}.
	\end{align*}
	For $k>1$ we can use our observations in the proof of Lemma~\ref{lemma:EXk} to come up with the following:
	\begin{align*}
	\mathbb{E}[X_k]&=\frac{e^2\binom{n}{2}^2}{k^2(k-1)}\cdot\left(-1+\frac{k-1+\alpha}{\alpha}\cdot\alpha^{k-1}e^\alpha\int_\alpha^\infty\frac{e^{-t}}{t^{k-1}},\mathrm{d}t\right)\\
	&\geq\frac{e^2\binom{n}{2}^2}{k^2(k-1)}\cdot\left(-1+\frac{k-1+\alpha}{\alpha}\cdot\frac{\alpha^2+(k+1)\alpha}{\alpha^2+2k\alpha+k(k-1)}\right)\\
	&=\frac{e^2\binom{n}{2}^2}{e^2\binom{n}{2}^2F_{n-k}^2+2k^2e\binom{n}{2}F_{n-k}+k^3(k-1)},
	\end{align*}
	where, as in Lemma~\ref{lemma:EXk}, we used $\alpha:=e\binom{n}{2}F_{n-k}/k$ to shorten notation.
	
	Next we need upper bounds for $\mathbb{E}[X_k^2]$. For $k=1$ we obtain upon direct computation that
	\begin{align*}
	\mathbb{E}\left[X_1^2\right]&=\int_0^\infty\frac{z_1(x)}{\left(F_{n-1}+x\right)^4}\,\mathrm{d}x=e\tbinom{n}{2}\cdot\int_0^\infty\frac{e^{-ex\binom{n}{2}}}{\left(F_{n-1}+x\right)^4}\,\mathrm{d}x\\
	&=e\tbinom{n}{2}\cdot\left(\frac{2-\alpha+\alpha^2}{6F_{n-1}^3}-\frac{e^2\binom{n}{2}^2}{6F_{n-1}}\cdot \alpha e^\alpha\int_\alpha^\infty\frac{e^{-t}}{t}\,\mathrm{d}t\right),
	\end{align*}
	where $\alpha:=e\binom{n}{2}F_{n-1}$ to shorten notation. Again applying the second Pad\'e approximant for the exponential integral in this expression \cite{Luke1969} (see the proof of Lemma~\ref{lemma:EXk}) yields
	\begin{align*}
	\mathbb{E}\left[X_1^2\right]&\leq e\tbinom{n}{2}\cdot\left(\frac{2-\alpha+\alpha^2}{6F_{n-1}^3}-\frac{e^2\binom{n}{2}^2}{6F_{n-1}}\cdot\frac{\alpha^2+3\alpha}{\alpha^2+4\alpha+2}\right)\\
	&=\frac{3e^2\binom{n}{2}^2F_{n-1}+2e\binom{n}{2}}{3e^2\binom{n}{2}^2F_{n-1}^5+12e\binom{n}{2}F_{n-1}^4+6F_{n-1}^3}.
	\end{align*}
	For $k=2$ we obtain upon direct computation that
	\begin{align*}
	\mathbb{E}\left[X_2^2\right]&=\int_0^\infty\frac{z_2(x)}{\left(F_{n-2}+x\right)^4}\,\mathrm{d}x=\frac{e^2\binom{n}{2}^2}{4}\cdot\int_0^\infty\frac{xe^{-ex\binom{n}{2}/2}}{\left(F_{n-2}+x\right)^4}\,\mathrm{d}x\\
	&=\frac{e^2\binom{n}{2}^2}{4}\cdot\left(\frac{1-2\alpha-\alpha^2}{6F_{n-2}^2}+\frac{e\binom{n}{2}(\alpha+3)}{12F_{n-2}}\cdot \alpha e^\alpha\int_\alpha^\infty\frac{e^{-t}}{t}\,\mathrm{d}t\right),
	\end{align*}
	where $\alpha:=e\binom{n}{2}F_{n-2}/2$ to shorten notation. The third Pad\'e approximant for the exponential integral in this expression \cite{Luke1969} states that for any $\alpha>0$ we have
	\begin{equation*}
	\alpha e^\alpha\int_\alpha^\infty\frac{e^{-t}}{t}\,\mathrm{d}t\leq\frac{\alpha^3+11\alpha^2+26\alpha+6}{\alpha^3+12\alpha^2+36\alpha+24}.
	\end{equation*}
	Applying this inequality yields
	\begin{align*}
	\mathbb{E}\left[X_2^2\right]&\leq\frac{e^2\binom{n}{2}^2}{4}\cdot\left(\frac{1-2\alpha-\alpha^2}{6F_{n-2}^2}+\frac{e\binom{n}{2}(\alpha+3)}{12F_{n-2}}\cdot\frac{\alpha^3+11\alpha^2+26\alpha+6}{\alpha^3+12\alpha^2+36\alpha+24}\right)\\
	&=\frac{e^3\binom{n}{2}^3F_{n-2}+8e^2\binom{n}{2}^2}{e^3\binom{n}{2}^3F_{n-2}^5+24e^2\binom{n}{2}^2F_{n-2}^4+144e\binom{n}{2}F_{n-2}^3+192F_{n-2}^2}.
	\end{align*}
	For $k=3$ we obtain upon direct computation that
	\begin{align*}
	\mathbb{E}\left[X_3^2\right]&=\int_0^\infty\frac{z_3(x)}{\left(F_{n-3}+x\right)^4}\,\mathrm{d}x=\frac{e^3\binom{n}{2}^3}{54}\cdot\int_0^\infty\frac{x^2e^{-ex\binom{n}{2}/3}}{\left(F_{n-3}+x\right)^4}\,\mathrm{d}x\\
	&=\frac{e^3\binom{n}{2}^3}{54}\cdot\left(\frac{2+5\alpha+\alpha^2}{6F_{n-3}}+\frac{\alpha^2+6\alpha+6}{6F_{n-3}}\cdot \alpha e^\alpha\int_\alpha^\infty\frac{e^{-t}}{t}\,\mathrm{d}t\right),
	\end{align*}
	where $\alpha:=e\binom{n}{2}F_{n-3}/3$ to shorten notation. The fourth Pad\'e approximant for the exponential integral in this expression \cite{Luke1969} states that for any $\alpha>0$ we have
	\begin{equation*}
	\alpha e^\alpha\int_\alpha^\infty\frac{e^{-t}}{t}\,\mathrm{d}t\geq\frac{\alpha^4+15\alpha^3+58\alpha^2+50\alpha}{\alpha^4+16\alpha^3+72\alpha^2+96\alpha+24}.
	\end{equation*}
	Applying this inequality yields
	\begin{align*}
	\mathbb{E}\left[X_3^2\right]&\leq\frac{e^3\binom{n}{2}^3}{54}\cdot\left(\frac{2+5\alpha+\alpha^2}{6F_{n-3}}+\frac{\alpha^2+6\alpha+6}{6F_{n-3}}\cdot\frac{\alpha^4+15\alpha^3+58\alpha^2+50\alpha}{\alpha^4+16\alpha^3+72\alpha^2+96\alpha+24}\right)\\
	&=\frac{e^4\binom{n}{2}^4F_{n-3}+12e^3\binom{n}{2}^3}{e^4\binom{n}{2}^4F_{n-3}^5+48e^3\binom{n}{2}^3F_{n-3}^4+648e^2\binom{n}{2}^2F_{n-3}^3+2592e\binom{n}{2}F_{n-3}^2+1944F_{n-3}}.
	\end{align*}
	For $k>3$ we obtain upon direct computation that
	\begin{align*}
	\mathbb{E}\left[X_k^2\right]&=\int_0^\infty\frac{z_k(x)}{\left(F_{n-k}+x\right)^4}\,\mathrm{d}x=\frac{e^k\binom{n}{2}^k}{k^k(k-1)!}\cdot\int_0^\infty\frac{x^{k-1}e^{-ex\binom{n}{2}/k}}{\left(F_{n-k}+x\right)^4}\,\mathrm{d}x\\
	&=\frac{e^4\binom{n}{2}^4}{6k^4(k-1)(k-2)(k-3)}\cdot\bigg(2k-3-(\alpha+k)^2\\
	&\qquad\qquad+\frac{\alpha^3+3(k-1)\alpha^2+3(k-1)(k-2)\alpha+(k-1)(k-2)(k-3)}{\alpha}\\
	&\qquad\qquad\qquad\qquad\cdot\alpha^{k-3}e^\alpha\int_\alpha^\infty\frac{e^{-t}}{t^{k-3}}\,\mathrm{d}t\bigg),
	\end{align*}
	where $\alpha:=e\binom{n}{2}F_{n-k}/k$ to shorten notation. The fourth Pad\'e approximant for the generalized exponential integral in this expression \cite{Luke1969} states that for any $\alpha>0$ and $k>3$ we have
	\begin{equation*}
	\alpha^{k-3}e^\alpha\int_\alpha^\infty\frac{e^{-t}}{t^{k-3}}\,\mathrm{d}t\leq\frac{\alpha^4+(3k+7)\alpha^3+3(k^2+3k+6)\alpha^2+(k+3)(k^2-k+10)\alpha+24}{\alpha^4+4(k+1)\alpha^3+6k(k+1)\alpha^2+4k(k^2-1)\alpha+k(k^2-1)(k-2)}.
	\end{equation*}
	Applying this inequality and immediately simplifying yields
	\begin{align*}
	&\mathbb{E}\left[X_k^2\right]\leq\left.\left(e^4\tbinom{n}{2}^4F_{n-k}+4ke^3\tbinom{n}{2}^3\right)\bigg/\left(e^4\tbinom{n}{2}^4F_{n-k}^5+4k(k+1)e^3\tbinom{n}{2}^3F_{n-k}^4\right.\right.\\
	&\qquad\qquad\qquad\left.+6k^3(k+1)e^2\tbinom{n}{2}^2F_{n-k}^3+4k^4(k^2-1)e\tbinom{n}{2}F_{n-k}^2+k^5(k^2-1)(k-2)F_{n-k}\right).
	\end{align*}
	Now observe that the bounds that we computed for $k=2$ and $k=3$ are actually the same as this bound. So, from now on we can use this last bound for any integer $k>1$.
	
	It remains now to combine the bounds that we just derived. For $k=1$ this yields
	\begin{align*}
	\mathbb{E}\left[X_1^2\right]-\left(\mathbb{E}\left[X_1\right]\right)^2&\leq\frac{3e^2\binom{n}{2}^2F_{n-1}+2e\binom{n}{2}}{3e^2\binom{n}{2}^2F_{n-1}^5+12e\binom{n}{2}F_{n-1}^4+6F_{n-1}^3}-\left(\frac{e\binom{n}{2}}{e\binom{n}{2}F_{n-1}^2+2F_{n-1}}\right)^2\\
	&=\frac{2e^3\binom{n}{2}^3F_{n-1}^2+14e^2\binom{n}{2}^2F_{n-1}+8e\binom{n}{2}}{3e^4\binom{n}{2}^4F_{n-1}^7+24e^3\binom{n}{2}^3F_{n-1}^6+66e^2\binom{n}{2}^2F_{n-1}^5+72e\binom{n}{2}F_{n-1}^4+24F_{n-1}^3}\\
	&=O\left(\frac{n^6F_{n-1}^2+n^2}{n^8F_{n-1}^7+F_{n-1}^3}\right).
	\end{align*}
	For $k>1$ we obtain
	\begin{align*}
	&\mathbb{E}\left[X_k^2\right]-\left(\mathbb{E}\left[X_k\right]\right)^2\\
	&\qquad\qquad\qquad\leq\left.\left(e^4\tbinom{n}{2}^4F_{n-k}+4ke^3\tbinom{n}{2}^3\right)\bigg/\left(e^4\tbinom{n}{2}^4F_{n-k}^5+4k(k+1)e^3\tbinom{n}{2}^3F_{n-k}^4\right.\right.\\
	&\qquad\qquad\qquad\qquad\left.+6k^3(k+1)e^2\tbinom{n}{2}^2F_{n-k}^3+4k^4(k^2-1)e\tbinom{n}{2}F_{n-k}^2+k^5(k^2-1)(k-2)F_{n-k}\right)\\
	&\qquad\qquad\qquad\qquad\qquad-\left(\frac{e^2\binom{n}{2}^2}{e^2\binom{n}{2}^2F_{n-k}^2+2k^2e\binom{n}{2}F_{n-k}+k^3(k-1)}\right)^2\\
	&\qquad\qquad\qquad=\frac{8k^3e^6\binom{n}{2}^6F_{n-k}^3+(...)+4k^7(k-2)^2e^3\binom{n}{2}^3}{e^8\binom{n}{2}^8F_{n-k}^9+(...)+k^{11}(k-1)^2(k^2-1)(k-2)F_{n-k}}\\
	&\qquad\qquad\qquad=O\left(\frac{k^3n^{12}F_{n-k}^3+k^9n^6}{n^{16}F_{n-k}^9+k^{16}F_{n-k}}\right),
	\end{align*}
	where the last inequality holds for $k>2$. For $k=2$ the last term in the denominator vanishes, which leads to the following result:
	\begin{equation*}
	\mathbb{E}\left[X_2^2\right]-\left(\mathbb{E}\left[X_2\right]\right)^2\leq O\left(\frac{n^{12}F_{n-2}^3+n^6}{n^{16}F_{n-2}^9+n^2F_{n-2}^2}\right)=O\left(\frac{n^{10}F_{n-2}^3+n^4}{n^{14}F_{n-2}^9+F_{n-2}^2}\right),
	\end{equation*}
	which finishes this proof.\blokje
\end{proof}
\begin{proof}[Proof of Theorem \ref{thm:kappalarge}]
	Using the Cauchy-Schwarz inequality for random variables (see Section~\ref{sect:prelim}), we obtain
	\begin{equation*}
	\mathbb{E}\left[\frac{\mathsf{ALG}}{\mathsf{OPT}}\right]\leq\sqrt{\mathbb{E}\left[\mathsf{ALG}^2\right]}\cdot\sqrt{\mathbb{E}\left[\frac1{\mathsf{OPT}^2}\right]}.
	\end{equation*}
	Recall from Section~\ref{sect:heur} that we know the distribution of $\mathsf{ALG}$. We can use this to compute and bound $\mathbb{E}[\mathsf{ALG}^2]$. If $\kappa<n$, then we obtain
	\begin{equation*}
	\mathbb{E}\left[\mathsf{ALG}^2\right]=\left(F_\kappa+H_{n-1}-H_{\kappa-1}\right)^2+\sum_{i=\kappa}^{n-1}\frac1{i^2}=\left(F_\kappa+\ln(n/\kappa)+\Theta(1)\right)^2+\sum_{i=\kappa}^{n-1}\frac1{i^2}
	\end{equation*}
	which is $O(1)$ since $\kappa\in\Theta(n)$ and $F_\kappa\leq\kappa f_\kappa<\kappa/(\kappa-1)\leq2$ for such $\kappa$. If $\kappa=n$, then we have $\mathbb{E}[\mathsf{ALG}^2]=F_n^2$.\\
	It remains to bound $\mathbb{E}[1/\mathsf{OPT}^2]$. We start by using our final notion from Section~\ref{sect:modelnotation}, and subsequently using the result of Lemma~\ref{lemma:LBOPTnk}. This yields
	\begin{equation*}
	\mathbb{E}\left[\frac1{\mathsf{OPT}^2}\right]=\mathbb{E}\left[\max_k\frac1{\mathsf{OPT}_{n-k}^2}\right]\leq\mathbb{E}\left[\max_k\frac1{\left(F_{n-k}+Z_k\right)^2}\right],
	\end{equation*}
	where $Z_k\sim\Gamma(k,e\binom{n}{2}/k)$ and where we take the maximum over $k\in\{0,\ldots,n-1\}$. Next we use the result of Lemma~\ref{lemma:expecmax} to get the maximum operator out of the expectation. This yields
	\begin{align*}
	\mathbb{E}\left[\frac1{\mathsf{OPT}^2}\right]&\leq\max_k\mathbb{E}\left[\frac1{\left(F_{n-k}+Z_k\right)^2}\right]+\sqrt{\frac{n-1}n\cdot\sum_{k=0}^{n-1}\Var\left(\frac1{\left(F_{n-k}+Z_k\right)^2}\right)}\\
	&\leq\max_k\mathbb{E}\left[X_k\right]+\sqrt{\sum_{k=0}^{n-1}\left(\mathbb{E}\left[X_k^2\right]-\left(\mathbb{E}\left[X_k\right]\right)^2\right)},
	\end{align*}
	where we also used $X_k:=1/(F_{n-k}+Z_k)^2$ to shorten notation, applied the difference formula for the variance, and used the inequality $(n-1)/n\leq1$.\\
	Since we know the distribution of $Z_k$, we can compute and subsequently bound the expectations of $X_k$ that occur in this last expression. For $k=0$ we have $Z_0=0$, and thus $\mathbb{E}[X_0]=1/F_n^2$ and $\mathbb{E}[X_0^2]-(\mathbb{E}[X_0])^2=0$. For $k\in[n-1]$, Lemmas~\ref{lemma:EXk} and \ref{lemma:VarXk} yield the bounds that we need to obtain the desired result.\blokje
\end{proof}
Finally, we will evaluate the just proven bound for the approximation ratio for the special case where all facility opening costs are equal, i.e., $f_1=\ldots=f_n=f$.
\begin{corollary}
	\label{cor:kappalarge}
	Assume that $f_1=\ldots=f_n=f$. Define $\kappa:=\kappa(n;f_1,\ldots,f_n)=\max\{i\in[n]:f_i<1/(i-1)\}=\min\{\lceil1/f\rceil,n\}$ and assume that $\kappa\in\Theta(n)$. Let $\mathsf{ALG}$ denote the total cost of the solution which opens, independently of the metric space, $\kappa$ arbitrarily chosen facilities, e.g., the facilities $\{1,\ldots,\kappa\}$. Then, it follows that
	\begin{equation*}
	\mathbb{E}\left[\frac{\mathsf{ALG}}{\mathsf{OPT}}\right]=O(1)+O\left(\sqrt[4]{\ln(n)n^3f^3}\right),
	\end{equation*}
	which for $f\in O(1/n\sqrt[3]{\ln(n)})$ is equal to $O(1)$.	Moreover, if $f\in o(1/n^3)$, then this approximation ratio becomes $1+o(1)$.
\end{corollary}
Before we can prove this corollary, we need two more lemmas.
\begin{lemma}
	\label{lemma:spmaxbound}
	Suppose that $f\in O(1/n)$. Then, for any $k\in[n-1]$ we have
	\begin{equation*}
	O\left(\frac{n^4(n-k)f+kn^2}{n^4(n-k)^3f^3+k^4(n-k)f}\right)\leq O\left(\frac1{n^2f^2}\right).
	\end{equation*}
	Moreover, if $f\in o(1/n^3)$, then for any $k\in[n-1]$ we have
	\begin{equation*}
	O\left(\frac{n^4(n-k)f+kn^2}{n^4(n-k)^3f^3+k^4(n-k)f}\right)\leq o\left(\frac1{n^2f^2}\right).
	\end{equation*}
\end{lemma}
\begin{proof}
	We consider three main intervals for $k$: $k\in o(n)$, $k=cn$ for some constant $c\in(0,1)$, and $k$ such that $n-k\in o(n)$.
	
	If $k\in o(n)$, then we have $n-k\in\Theta(n)$, and therefore it follows that
	\begin{equation*}
	O\left(\frac{n^4(n-k)f+kn^2}{n^4(n-k)^3f^3+k^4(n-k)f}\right)=O\left(\frac{n^5f+kn^2}{n^7f^3+k^4nf}\right)=O\left(\frac{n}{f}\cdot\frac{n^3f+k}{n^6f^2+k^4}\right)
	\end{equation*}
	in this case. Now, if $k\in O(\sqrt{n^3f})$ this last expression reduces to $O(n^4f/n^6f^3)=O(1/n^2f^2)$. On the other hand, if $k\in\Omega(n^3f)$ this last expression reduces to $O(kn/k^4f)=O(n/k^3f)\leq O(n/kf)\leq O(n/n^3f^2)=O(1/n^2f^2)$, where the first inequality follows since $k\geq1$ and the second since $k\in\Omega(n^3f)$. In the remaining cases, where $k\in\Omega(\sqrt{n^3f})\cap O(n^3f)$, the last expression reduces to $O(n^4f/k^4f)=O(n^4/k^4)\leq O(n^4/n^6f^2)=O(1/n^2f^2)$, where the inequality follows since $k\in\Omega(\sqrt{n^3f})$.\\
	If additionally $f\in o(1/n^3)$, then it follows that $n^3f\in o(1)$, and thus for any $k\in[n-1]$ with $k\in o(n)$ the expression reduces to $O(kn/k^4f)=O(n/k^3f)\leq O(n/kf)\leq o(n/n^3f^2)=o(1/n^2f^2)$, where the first inequality follows since $k\geq1$ and the second since $k\in\omega(n^3f)$.
	
	If $k=cn$ for some constant $c\in(0,1)$, then we have $k\in\Theta(n)$ and $n-k\in\Theta(n)$, and therefore it follows that
	\begin{equation*}
	O\left(\frac{n^4(n-k)f+kn^2}{n^4(n-k)^3f^3+k^4(n-k)f}\right)=O\left(\frac{n^5f+n^3}{n^7f^3+n^5f}\right)=O\left(\frac{1}{n^2f}\cdot\frac{n^2f+1}{n^2f^2+1}\right)
	\end{equation*}
	in this case. Now, if $f\in O(1/n^2)$ this last expression reduces to $O(1/n^2f)\leq o(1/n^2f^2)$, where the inequality follows since $f\in o(1)$. On the other hand, if $f\in\Omega(1/n^2)$ this last expression reduces to $O(n^2f/n^2f)=O(1)\leq O(1/n^2f^2)$, where the inequality follows since we still have $f\in O(1/n)$.
	
	If $k$ is such that $g:=n-k\in o(n)$, then we have $k\in\Theta(n)$, and therefore it follows that
	\begin{equation*}
	O\left(\!\frac{n^4(n-k)f+kn^2}{n^4(n-k)^3f^3+k^4(n-k)f}\!\right)=O\left(\!\frac{n^4gf+n^3}{n^4g^3f^3+n^4gf}\!\right)=O\left(\!\frac{1}{ngf}\cdot\frac{ngf+1}{g^2f^2+1}\right)
	\end{equation*}
	in this case. Now, if $g\in O(1/nf)$ this last expression reduces to $O(1/ngf)\leq O(g/nf)\leq O(1/n^2f^2)$, where the first inequality follows since $g\geq1$ and the second since $g\in O(1/nf)$. On the other hand, if $g\in\Omega(1/nf)$ this last expression reduces to $O(ngf/ngf)=O(1)\leq O(1/n^2f^2)$, where the inequality follows since we still have $f\in O(1/n)$.\\
	If additionally $f\in o(1/n^3)$, then it follows that $1/nf\in \omega(n^2)$, and thus for any $k\in[n-1]$ with $g=n-k\in o(n)$ the expression reduces to $O(1/ngf)\leq O(g/nf)\leq o(1/n^2f^2)$, where the first inequality follows since $g\geq1$ and the second since $g\in o(1/nf)$.
	
	Since any $k\in[n-1]$ belongs in one of the main intervals that we've investigated above, the proof is now complete.\blokje
\end{proof}
\begin{lemma}
	\label{lemma:spsumbound}
	Suppose that $f\in O(1/n)$. Then we have
	\begin{align*}
	&O\left(\frac{n^8f^2+n^2}{n^{15}f^7+n^3f^3}+\frac{n^{13}f^3+n^4}{n^{25}f^9+n^2f^2}+\sum_{k=3}^{n-1}\frac{k^3n^{12}(n-k)^3f^3+k^9n^6}{n^{16}(n-k)^9f^9+k^{16}(n-k)f}\right)\\
	&\qquad\qquad\qquad\leq O\left(\frac{\ln(n)}{nf}+\frac1{n^4f^4}\right).
	\end{align*}
	Moreover, if $f\in o(1/n^3)$ then this result can be improved to $O(1/nf^3)$.
\end{lemma}
\begin{proof}
	The first two terms of this summation (corresponding to $k=1$ and $k=2$) are slightly different and will be considered at the end of this proof. For the remaining terms of the summation, we will consider three different intervals. Define $\mathcal{A}:=\{k\in\{3,\ldots,n-1\}:k\in o(n)\}$, $\mathcal{C}:=\{k\in\{3,\ldots,n-1\}:n-k\in o(n)\}$, and
	$\mathcal{B}:=\{3,\ldots,n-1\}\backslash(\mathcal{A}\cup\mathcal{C})$. Note that $\mathcal{A}\cup\mathcal{B}\cup\mathcal{C}=\{3,\ldots,n-1\}$ by these definitions and that for any $k\in\mathcal{B}$ we have $k\in\Theta(n)$ and $n-k\in\Theta(n)$.
	
	For any $k\in\mathcal{C}$ we have $k\in\Theta(n)$, and therefore it follows that
	\begin{align*}
	O\left(\frac{k^3n^{12}(n-k)^3f^3+k^9n^6}{n^{16}(n-k)^9f^9+k^{16}(n-k)f}\right)&=O\left(\frac{n^{15}(n-k)^3f^3+n^{15}}{n^{16}(n-k)^9f^9+n^{16}(n-k)f}\right)\\
	&=O\left(\frac{1}{n(n-k)f}\cdot\frac{(n-k)^3f^3+1}{(n-k)^8f^8+1}\right)=O\left(\frac{1}{n(n-k)f}\right),
	\end{align*}
	where the last equality follows since $f\in O(1/n)$ and $k\in\mathcal{C}$ implies that $(n-k)f\in o(1)$. We also have
	\begin{equation*}
	\sum_{k\in\mathcal{C}}O\left(\frac1{n(n-k)f}\right)=\sum_{g:n-g\in\mathcal{C}}O\left(\frac1{ngf}\right)\leq\sum_{g\in[n-1]}O\left(\frac1{ngf}\right)=O\left(\frac{\ln(n)}{nf}\right).
	\end{equation*}
	
	For any $k\in\mathcal{B}$ we have $k\in\Theta(n)$ and $n-k\in\Theta(n)$, and therefore it follows that
	\begin{align*}
	O\left(\frac{k^3n^{12}(n-k)^3f^3+k^9n^6}{n^{16}(n-k)^9f^9+k^{16}(n-k)f}\right)&=O\left(\frac{n^{18}f^3+n^{15}}{n^{25}f^9+n^{17}f}\right)\\
	&=O\left(\frac{1}{n^2f}\cdot\frac{n^3f^3+1}{n^8f^8+1}\right)=O\left(\frac{1}{n^2f}\right),
	\end{align*}
	where the last equality follows since $f\in O(1/n)$ implies that $nf\in O(1)$. We also have
	\begin{equation*}
	\sum_{k\in\mathcal{B}}O\left(\frac1{n^2f}\right)=O\left(\frac{|\mathcal{B}|}{n^2f}\right)=O\left(\frac1{nf}\right).
	\end{equation*}
	
	Finally, for any $k\in\mathcal{A}$ we have $n-k\in\Theta(n)$, and therefore it follows that
	\begin{align*}
	O\left(\frac{k^3n^{12}(n-k)^3f^3+k^9n^6}{n^{16}(n-k)^9f^9+k^{16}(n-k)f}\right)&=O\left(\frac{k^3n^{15}f^3+k^9n^6}{n^{25}f^9+k^{16}nf}\right)\\
	&=O\left(\frac{k^3n^5}{f}\cdot\frac{n^9f^3+k^6}{n^{24}f^8+k^{16}}\right).
	\end{align*}
	We split $\mathcal{A}$ into two parts: $\mathcal{A}_1:=\{k\in\mathcal{A}:k^2\leq n^3f\}$ and $\mathcal{A}_2:=\{k\in\mathcal{A}:k^2>n^3f\}$. Note that $\mathcal{A}=\mathcal{A}_1\cup\mathcal{A}_2$ and that $\mathcal{A}_1=\varnothing$ for sufficiently large $n$ if $f\in o(1/n^3)$. If $k\in\mathcal{A}_1$ then our last expression becomes $O(k^3n^{14}f^3/n^{24}f^9)=O(k^3/n^{10}f^6)$. On the other hand, if $k\in\mathcal{A}_2$ then it becomes $O(k^9n^5/k^{16}f)=O(n^5/k^7f)$. Summing these values yields
	\begin{equation*}
	\sum_{k\in\mathcal{A}_1}O\left(\frac{k^3}{n^{10}f^6}\right)=O\left(\frac{n^6f^2}{n^{10}f^6}\right)=O\left(\frac1{n^4f^4}\right)
	\end{equation*}
	and
	\begin{equation*}
	\sum_{k\in\mathcal{A}_2}O\left(\frac{n^5}{k^7f}\right)=O\left(\frac{n^5}{n^6f}\right)=O\left(\frac1{nf}\right).
	\end{equation*}
	
	Now only the cases $k=1$ and $k=2$ remain. For $k=2$ we have
	\begin{equation*}
	O\left(\frac{n^{11}f^3+n^2}{n^{21}f^9+f^2}\right)=O\left(\frac{n^2}{f^2}\cdot\frac{n^9f^3+1}{n^{21}f^7+1}\right)=\left\{\begin{array}{ll}O(n^2/f^2)&\text{if }f\in o(1/n^3),\\O(1/n^{10}f^6)&\text{if }f\in\Omega(1/n^3),\end{array}\right.
	\end{equation*}
	whereas for $k=1$ we have
	\begin{equation*}
	O\left(\frac{n^6f^2+1}{n^{13}f^7+nf^3}\right)=O\left(\frac{1}{nf^3}\cdot\frac{n^6f^2+1}{n^{12}f^4+1}\right)=\left\{\begin{array}{ll}O(1/nf^3)&\text{if }f\in o(1/n^3),\\O(1/n^7f^5)&\text{if }f\in\Omega(1/n^3).\end{array}\right.
	\end{equation*}
	
	Combining everything together for $f\in\Omega(1/n^3)$, we obtain
	\begin{align*}
	&O\left(\frac{n^6f^2+1}{n^{13}f^7+nf^3}\right)\!+O\left(\frac{n^{11}f^3+n^2}{n^{21}f^9+f^2}\right)+\sum_{k=3}^{n-1}O\left(\frac{k^3n^{12}(n-k)^3f^3+k^9n^6}{n^{16}(n-k)^9f^9+k^{16}(n-k)f}\right)\\
	&\qquad\qquad\qquad=O\left(\frac1{n^7f^5}+\frac1{n^{10}f^6}+\frac1{n^4f^4}+\frac1{nf}+\frac1{nf}+\frac{\ln(n)}{nf}\right)=O\left(\frac1{n^4f^4}+\frac{\ln(n)}{nf}\right),
	\end{align*}
	whereas for $f\in o(1/n^3)$ we obtain
	\begin{align*}
	&O\left(\frac{n^6f^2+1}{n^{13}f^7+nf^3}\right)\!+O\left(\frac{n^{11}f^3+n^2}{n^{21}f^9+f^2}\right)+\sum_{k=3}^{n-1}O\left(\frac{k^3n^{12}(n-k)^3f^3+k^9n^6}{n^{16}(n-k)^9f^9+k^{16}(n-k)f}\right)\\
	&\qquad\qquad\qquad=O\left(\frac1{nf^3}+\frac{n^2}{f^2}+0+\frac1{nf}+\frac1{nf}+\frac{\ln(n)}{nf}\right)=O\left(\frac1{nf^3}\right),
	\end{align*}
	which finishes this proof.\blokje
\end{proof}
\begin{proof}[Proof of Corollary~\ref{cor:kappalarge}]
	Observe that $\kappa\in\Theta(n)$ and $\kappa=\min\{\lceil1/f\rceil,n\}$ implies that $f\in O(1/n)$. 
	We start by bounding the maximum in the first term. Lemma~\ref{lemma:spmaxbound} shows that in our special case this maximum is asymptotically bounded by the first element. Using this result, we can now bound the maximum in the first term by $O(1/n^2f^2)$. Moreover, if $f\in o(1/n^3)$, then for sufficiently large $n$ it follows that the maximum is given by its first element, i.e., it is equal to $1/n^2f^2$.
	
	Next, we evaluate the sum of the variances. Lemma~\ref{lemma:spsumbound} provides the corresponding result. If $f\in\Theta(1/n)$, then it follows that
	\begin{equation*}
	\mathbb{E}\left[\frac{\mathsf{ALG}}{\mathsf{OPT}}\right]\leq\sqrt{O\left(\frac1{n^2f^2}\right)}+\sqrt[4]{O\left(\frac{\ln(n)}{nf}+\frac1{n^4f^4}\right)}=O\left(\sqrt[4]{\ln(n)}\right).
	\end{equation*}
	If $f\in o(1/n)$ then we have for sufficiently large $n$ that $\kappa=n$ and therefore
	\begin{equation*}
	\mathbb{E}\left[\frac{\mathsf{ALG}}{\mathsf{OPT}}\right]\leq \sqrt{O\left(\frac{n^2f^2}{n^2f^2}\right)}+\sqrt[4]{O\left(\frac{\ln(n)n^4f^4}{nf}+\frac{n^4f^4}{n^4f^4}\right)}=O\left(1+\sqrt[4]{\ln(n)n^3f^3}\right),
	\end{equation*}
	where the last term in general is bounded by $O(\sqrt[4]{\ln(n)})$ (since $f\in O(1/n)$) and more specifically by $O(1)$ if $f\in O(1/n\sqrt[3]{\ln(n)})$. If $f\in o(1/n^3)$, then we obtain
	\begin{equation*}
	\mathbb{E}\left[\frac{\mathsf{ALG}}{\mathsf{OPT}}\right]\leq nf\left(\sqrt{\frac1{n^2f^2}}+\sqrt[4]{O\left(\frac1{nf^3}\right)}\right)=1+O\left(\sqrt[4]{n^3f}\right)=1+o(1),
	\end{equation*}
	which finishes this proof.\blokje
\end{proof}

\section{Concluding Remarks}\label{sect:final}

We have analyzed a rather simple heuristic for the (uncapacitated) facility location problem on random shortest path metrics. We have shown that in many cases this heuristic produces a solution which is surprisingly close to the optimal solution as the size of the instances grows. A logical next step would be to look at heuristics that are (slightly) more sophisticated, and see whether their performance on random shortest path metrics is better than our simple heuristic.

On the other hand there are many other $\mathcal{NP}$-hard (combinatorial) optimization problems for which it would be interesting to know how they behave on random short path metrics.

\bibliographystyle{abbrvnat}
\bibliography{References}

\begin{thebibliography}{18}
\providecommand{\natexlab}[1]{#1}
\providecommand{\url}[1]{\texttt{#1}}
\expandafter\ifx\csname urlstyle\endcsname\relax
  \providecommand{\doi}[1]{doi: #1}\else
  \providecommand{\doi}{doi: \begingroup \urlstyle{rm}\Url}\fi

\bibitem[Ahn et~al.(1988)Ahn, Cooper, {Cornu\'ejols}, and Frieze]{Ahn1988}
S.~Ahn, C.~Cooper, G.~{Cornu\'ejols}, and A.~Frieze.
\newblock Probabilistic analysis of a relaxation for the k-median problem.
\newblock \emph{Mathematics of Operations Research}, 13\penalty0 (1):\penalty0
  1--31, 1988.
\newblock \doi{10.1287/moor.13.1.1}.

\bibitem[Aven(1985)]{Aven1985}
T.~Aven.
\newblock Upper (lower) bounds on the mean of the maximum (minimum) of a number
  of random variables.
\newblock \emph{Journal of Applied Probability}, 22\penalty0 (3):\penalty0
  723--728, 1985.
\newblock \doi{10.2307/3213876}.

\bibitem[Bringmann et~al.(2015)Bringmann, Engels, Manthey, and
  Rao]{Bringmann2015}
K.~Bringmann, C.~Engels, B.~Manthey, and B.~V.~R. Rao.
\newblock Random shortest paths: Non-euclidean instances for metric
  optimization problems.
\newblock \emph{Algorithmica}, 73\penalty0 (1):\penalty0 42--62, 2015.
\newblock \doi{10.1007/s00453-014-9901-9}.

\bibitem[Cornuejols et~al.(1990)Cornuejols, Nemhauser, and
  Wolsey]{Cornuejols1990}
G.~Cornuejols, G.~L. Nemhauser, and L.~A. Wolsey.
\newblock The uncapacitated facility location problem.
\newblock In P.~B. Mirchandani and R.~L. Francis, editors, \emph{Discrete
  Location Theory}, chapter~3, pages 119--171. Wiley-Interscience, New York,
  1990.
\newblock ISBN 978-0-471-89233-5.

\bibitem[Davis and Prieditis(1993)]{Davis1993}
R.~Davis and A.~Prieditis.
\newblock The expected length of a shortest path.
\newblock \emph{Information Processing Letters}, 46\penalty0 (3):\penalty0
  135--141, 1993.
\newblock \doi{10.1016/0020-0190(93)90059-I}.

\bibitem[Flaxman et~al.(2007)Flaxman, Frieze, and Vera]{Flaxman2007}
A.~D. Flaxman, A.~M. Frieze, and J.~C. Vera.
\newblock On the average case performance of some greedy approximation
  algorithms for the uncapacitated facility location problem.
\newblock \emph{Combinatorics, Probability and Computing}, 16\penalty0
  (5):\penalty0 713--732, 2007.
\newblock \doi{10.1017/S096354830600798X}.

\bibitem[Frieze and Yukich(2007)]{Frieze2007}
A.~M. Frieze and J.~E. Yukich.
\newblock Probabilistic analysis of the {TSP}.
\newblock In G.~Gutin and A.~P. Punnen, editors, \emph{The Traveling Salesman
  Problem and Its Variations}, chapter~7, pages 257--307. Springer, Boston, MA,
  2007.
\newblock \doi{10.1007/0-306-48213-4_7}.

\bibitem[Hammersley and Welsh(1965)]{Hammersley1965}
J.~M. Hammersley and D.~J.~A. Welsh.
\newblock First-passage percolation, subadditive processes, stochastic
  networks, and generalized renewal theory.
\newblock In J.~Neyman and L.~M. {Le Cam}, editors, \emph{Bernoulli 1713 Bayes
  1763 Laplace 1813}, pages 61--110. Springer Berlin Heidelberg, 1965.
\newblock \doi{10.1007/978-3-642-49750-6_7}.

\bibitem[Hassin and Zemel(1985)]{Hassin1985}
R.~Hassin and E.~Zemel.
\newblock On shortest paths in graphs with random weights.
\newblock \emph{Mathematics of Operations Research}, 10\penalty0 (4):\penalty0
  557--564, 1985.
\newblock \doi{10.1287/moor.10.4.557}.

\bibitem[Howard(2004)]{Howard2004}
C.~D. Howard.
\newblock Models of first-passage percolation.
\newblock In H.~Kesten, editor, \emph{Probability on Discrete Structures},
  pages 125--173. Springer Berlin Heidelberg, 2004.
\newblock \doi{10.1007/978-3-662-09444-0_3}.

\bibitem[Janson(1999)]{Janson1999}
S.~Janson.
\newblock One, two and three times log n/n for paths in a complete graph with
  random weights.
\newblock \emph{Combinatorics, Probability and Computing}, 8\penalty0
  (4):\penalty0 347--361, 1999.
\newblock \doi{10.1017/S0963548399003892}.

\bibitem[Karp and Steele(1985)]{Karp1985}
R.~M. Karp and J.~M. Steele.
\newblock Probabilistic analysis of heuristics.
\newblock In E.~L. Lawler, J.~K. Lenstra, A.~H.~G. {Rinnooy Kan}, and D.~B.
  Shmoys, editors, \emph{The Traveling Salesman Problem: A Guided Tour of
  Combinatorial Optimization}, pages 181--205. John Wiley {\&} Sons Ltd., 1985.
\newblock ISBN 978-0-471-90413-7.

\bibitem[Li(2013)]{Li2011}
S.~Li.
\newblock A $1.488$ approximation algorithm for the uncapacitated facility
  location problem.
\newblock \emph{Information and Computation}, 222:\penalty0 45--58, 2013.
\newblock \doi{10.1016/j.ic.2012.01.007}.

\bibitem[Luke(1969)]{Luke1969}
Y.~L. Luke.
\newblock Chapter {XIV} polynomial and rational approximations for the
  incomplete gamma function.
\newblock In \emph{The Special Functions and their Approximations}, volume 53,
  part 2 of \emph{Mathematics in Science and Engineering}, pages 186--213.
  Elsevier, 1969.
\newblock ISBN 978-0-124-59902-4.
\newblock \doi{10.1016/S0076-5392(09)60074-6}.

\bibitem[Nagaraja(2006)]{Nagaraja2006}
H.~N. Nagaraja.
\newblock Order statistics from independent exponential random variables and
  the sum of the top order statistics.
\newblock In N.~Balakrishnan, J.~M. Sarabia, and E.~Castillo, editors,
  \emph{Advances in Distribution Theory, Order Statistics, and Inference},
  chapter~11, pages 173--185. Birkh{\"a}user Boston, 2006.
\newblock \doi{10.1007/0-8176-4487-3_11}.

\bibitem[{R\'enyi}(1953)]{Renyi1953}
A.~{R\'enyi}.
\newblock On the theory of order statistics.
\newblock \emph{Acta Mathematica Academiae Scientiarum Hungarica}, 4\penalty0
  (3-4):\penalty0 191--231, 1953.
\newblock \doi{10.1007/BF02127580}.

\bibitem[Ross(2010)]{Ross2010}
S.~M. Ross.
\newblock \emph{Introduction to Probability Models}.
\newblock Academic Press, Burlington, MA, 10th edition, 2010.
\newblock ISBN 978-0-12-375686-2.

\bibitem[Shaked and Shanthikumar(2007)]{Shaked2007}
M.~Shaked and J.~G. Shanthikumar.
\newblock \emph{Stochastic Orders}.
\newblock Springer, New York, NY, 2007.
\newblock ISBN 978-0-387-34675-5.
\newblock \doi{10.1007/978-0-387-34675-5}.

\end{thebibliography}
%
%
%
%
%

%
%
%
%
%
%
%
%
%
%
%
%
%
%
%
%
%
%
%
%
%
%
%
%
%
%

\end{document}